\documentclass[a4paper,11pt]{article}

\usepackage{jheppub} 
\usepackage[T1]{fontenc} 

\usepackage{verbatim}
\usepackage{enumerate}
\usepackage{mathrsfs}
\usepackage{bm}
\usepackage{tikz}

\newcommand{\tr}{\text{tr}}

\DeclareMathOperator*{\ext}{ext}

\newcommand{\Del}{\nabla}
\newcommand{\del}{\partial}
\newcommand{\dd}{\text{d}}

\renewcommand{\tilde}{\widetilde}
\renewcommand{\bar}{\overline}

\usetikzlibrary{arrows, positioning, shapes.geometric, decorations.markings, patterns}

\newtheorem{theorem}{Theorem}

\newtheorem{conjecture}[theorem]{Conjecture}
\newtheorem{corollary}[theorem]{Corollary}
\newtheorem{definition}[theorem]{Definition}
\newtheorem{lemma}[theorem]{Lemma}

\newenvironment{proof}[1][Proof]{\noindent\textbf{#1.} }{\ \rule{0.5em}{0.5em}}

\newcommand{\bra}[1]{\langle #1 |}
\newcommand{\ket}[1]{| #1 \rangle}
\newcommand{\braket}[2]{\langle #1 | #2 \rangle}

\title{Holographic entanglement entropy is cutoff-covariant}

\author{Jonathan Sorce}

\affiliation{Stanford Institute for Theoretical Physics, Stanford University, 382 Via Pueblo Mall, Stanford, CA 94305-4060, U.S.A.}

\emailAdd{jsorce@stanford.edu}

\abstract{In the context of the AdS/CFT correspondence, it is often convenient to regulate infinite quantities in asymptotically anti-de Sitter spacetimes by introducing a sharp cutoff at some finite, large value of a particular radial coordinate. This procedure is \emph{a priori} coordinate dependent, and may not be well-motivated in full, covariant general relativity; however, the fact that physically meaningful quantities such as the entanglement entropy can be obtained by such a regulation procedure suggests some underlying covariance. In this paper, we provide a careful treatment of the radial cutoff procedure for computing holographic entanglement entropy in asymptotically anti-de Sitter spacetimes. We prove two results that are frequently assumed in the literature, but that have not been carefully addressed: (i) that the choice of a ``globally minimal surface'' among several extremal candidates is independent of the choice of regulator, and (ii) that finite CFT quantities such as the mutual information which involve ``divergence-cancelling'' sums of entanglement entropies are well-defined under the usual prescription for computing covariant holographic entanglement entropy. Our results imply that the ``globally minimal surface'' prescription for computing the holographic entanglement entropy is well-posed from the perspective of general relativity, and thus support the widely-held belief that this is the correct prescription for identifying the entanglement wedge of a boundary subregion in AdS/CFT. We also comment on the geometric source of state-dependent divergences in the holographic entanglement entropy, and identify precisely the regime of validity of the ``vacuum subtraction'' protocol for regulating infinite entanglement entropies in arbitrary states by comparing them to the entanglement entropies of identical regions in the vacuum.
Our proofs make use of novel techniques for the covariant analysis of extremal surfaces, which are explained in detail and may find use more broadly in the study of holographic entanglement entropy.}

\makeatletter
\def\@fpheader{\vspace{0.2cm}}
\makeatother

\begin{document} 
\maketitle
\flushbottom

\section{Introduction}

In the AdS/CFT correspondence, it has been proposed \cite{RT, HRT, LM, DLR} that the entanglement entropies of subregions in particular conformal field theory (CFT) states correspond to the areas of extremal surfaces in particular asymptotically anti-de Sitter (AdS) spacetimes. More precisely, in a ``holographic'' CFT state with a semiclassical gravitational dual, the entanglement entropy $S(A)$ of a codimension-1, spacelike or null boundary region $A$ is computed to leading order in $G_N$ by
\begin{equation} \label{eq:HRT}
    S(A) = \ext_{\Sigma \sim A} \frac{\mathrm{Area}(\Sigma)}{4 G_N},
\end{equation}
where the extremum is taken over all codimension-2 bulk surfaces $\Sigma$ that are homologous to the boundary region $A$.\footnote{When we say $A$ and $\Sigma$ are ``homologous,'' we mean that there exists a codimension-1 hypersurface in the bulk that has $A \cup \Sigma$ as its boundary.} In particular, the homology constraint requires that the boundary of $\Sigma$ coincide with the entangling surface $\del A$. A surface $\Sigma$ that extremizes \eqref{eq:HRT} is called a Ryu-Takayanagi (RT) surface in a static spacetime, or more generally a Hubeny-Rangamani-Takayanagi (HRT) surface in a dynamical spacetime.

Calculations involving equation \eqref{eq:HRT} are troubled by the fact that both $S(A)$ and $\mathrm{Area}(\Sigma)$ are formally infinite when $\del A$ is nonempty. The entanglement entropy is infinite due to ultraviolet divergences in short-range correlations across the entangling surface, while the area of $\Sigma$ is infinite due to the fact that areas diverge near the boundary of an asymptotically anti-de Sitter spacetime. In some sense, this means that \eqref{eq:HRT} is trivially satisfied, as it gives the correct answer for the entanglement entropy: $\infty = \infty$. This result, however, is not especially satisfying; in order for \eqref{eq:HRT} to be meaningful, it must be useful for computing manifestly finite quantities in the CFT. For example, the \emph{mutual information} between regions $A$ and $B$,
\begin{equation} \label{eq:MI}
    I(A:B) = S(A) + S(B) - S(A\cup B),
\end{equation}
is generally finite when $A$ and $B$ are non-adjacent --- divergences in the entanglement entropies come from correlations across the entangling surfaces $\del A$, $\del B$, and $\del (A \cup B)$, and these divergences cancel in \eqref{eq:MI} when $\del A$ and $\del B$ are disjoint. Many other information-theoretic quantities of interest, such as the \emph{conditional mutual information} 
\begin{equation} \label{eq:CMI}
    I(A:B|C) = S(AC) + S(BC) - S(ABC) - S(C),
\end{equation}
involve such ``divergence-cancelling'' sums of entanglement entropies and are thus believed to be finite.

In order to compute these finite CFT quantities using the holographic dictionary, it is necessary to introduce a cutoff in equation \eqref{eq:HRT}. One can regulate the areas of surfaces by introducing a sharp cutoff in an arbitrary radial coordinate, perform the calculation using these regulated areas, and hope that the correct, finite CFT quantity is obtained in the limit as the cutoff is removed. \emph{A priori}, however, it is not obvious that different choices of radial coordinate for cutting off spacetime will give the same finite answers in the limit as the cutoff is removed, or even that an arbitrarily-chosen radial cutoff will give a finite answer at all. Since finite sums of entanglement entropies are well-defined in the CFT, the holographic entanglement entropy proposal \eqref{eq:HRT} can only hold if the finite answers one obtains by adding and subtracting the areas of radially regulated surfaces are independent of the chosen radial cutoff in the limit as the cutoff is removed.

Furthermore, there is an ambiguity in equation \eqref{eq:HRT} when there exist multiple extremal surfaces anchored to the same boundary region. The usual prescription for choosing the ``correct'' HRT surface, proposed in \cite{HRT}, is to choose the ``globally minimal surface'' among various extremal candidates by introducing a radial cutoff and asking which candidate surface has the smallest (finite) area. This prescription, again, is not \emph{a priori} coordinate-invariant; one could imagine that different radial coordinates could pick out different extremal surfaces as being ``globally minimal.'' Such an ambiguity would pose serious problems for the holographic dictionary: first because CFT quantities such as the mutual information depend on the choice of a ``correct'' extremal surface for computing the entanglement entropy, and second because the choice of extremal surface corresponding to a boundary subregion determines the bulk entanglement wedge of that subregion. If entanglement wedge reconstruction is to be believed \cite{DHW, noisyDHW}, then the entanglement wedge of a boundary region is uniquely determined by bulk reconstruction; it cannot be subject to the whims of our chosen coordinates.

Both of these issues can be resolved by one simple observation: that \emph{the finite difference in area between any two extremal surfaces homologous to the same boundary region is cutoff-independent}. The main contribution of this paper is to prove this statement. We claim first that for any reasonable choice of radial coordinate $r$, extremal surfaces with the same boundary anchor have identical area divergences in $r$; this implies that the area difference between any two extremal surfaces with the same boundary anchor is finite under any radial cutoff prescription. We then show that this finite difference is independent of the radial coordinate $r$ chosen to regulate the calculation. This result implies both (i) that the ``globally minimal surface'' among several extremal candidates is well-defined as the unique surface whose area difference with every other extremal candidate is negative, and (ii) that ``divergence-cancelling'' sums of extremal surface areas such as those appearing in equations \eqref{eq:MI} and \eqref{eq:CMI} are cutoff-independent. Our results hold not only for asymptotically AdS spacetimes, but also for \emph{asymptotically locally AdS} spacetimes that are proposed to have their own holographic correspondence (see, e.g., \cite{HM1998}). The machinery developed to prove this theorem will also allow us to provide a rigorous definition of the ``vacuum subtraction'' protocol for defining finite entanglement entropies; we will explain how to use a single regulator shared between two different spacetimes to subtract the vacuum entanglement entropy of a subregion from the entanglement entropy in a generic state, and identify the metric falloff conditions under which this subtraction procedure yields a consistent, finite answer.\footnote{While it is frequently assumed that all spacetimes of interest have universal, vacuum-like divergences in the entanglement entropy, there exist holographic spacetimes with state-dependent divergences; in these spacetimes, the vacuum subtraction protocol is ill-defined. We will comment further on these spacetimes, and on their relationship to the state-dependent divergences previously identified by Marolf and Wall \cite{MW}, in Section \ref{sec:two-spacetimes}.}

Finally, there has been some confusion in the literature as to how a boundary-anchored extremal surface should be transported to the location of the cutoff. One prescription is to impose the cutoff directly on the boundary-anchored extremal surface, while another dictates that one should first transport the boundary region $R$ to the cutoff surface along the radial direction, then compute the area of a bulk extremal surface homologous to this new, ``transported'' subregion of the cutoff surface. These two different prescriptions are sketched in Figure \ref{fig:two-cutoff-prescriptions-1} for the extremal geodesic of an interval in vacuum $AdS_3.$ A straightforward corollary of our results answers the question of which prescription is preferred: subject to a modest conjecture about the existence of extremal surfaces with ``transported'' boundary conditions, the choice does not matter --- both prescriptions give the same answer for the finite area difference between two extremal surfaces with the same boundary anchor. This corollary appears in Section \ref{sec:the-big-one}.

\begin{figure}
    \centering
	\begin{tikzpicture}[scale=1.5, thick]
	\draw (0, 0) circle (2.5cm);

	\draw [line width=2pt, lightgray!30!blue] (-1.7, 1.83) arc (42.89:-42.89:2.69);
    
    \draw [line width=2pt, black!15!lime!10!orange] (-1.16, 1.25) to [out=-60, in=60] (-1.16,-1.25);
    
    \draw [ultra thick, dashed] (-1.68, 1.85) to (-1.16, 1.25);
    \draw [ultra thick, dashed] (-1.68, -1.85) to (-1.16, -1.25);
    
	\draw [purple, ultra thick, dashed] (0,0) circle (1.7cm);
	\end{tikzpicture}
	
	\caption{Two different prescriptions for cutting off the area of an extremal surface, shown here in a constant-time slice of vacuum $AdS_3$ in global coordinates. In the first prescription, represented in blue, the radial cutoff is imposed directly on the boundary-anchored surface. In the second prescription, represented in orange, the boundary region is first transported to the cutoff along radial curves (dashed black lines), then a new extremal surface is found using the initial data of this ``transported'' boundary. In Section \ref{sec:the-big-one}, we show that the two prescriptions give identical answers for the finite difference in area between two extremal surfaces with the same boundary anchor, subject to a conjecture about the existence of cutoff-anchored extremal surfaces.}
	\label{fig:two-cutoff-prescriptions-1}
\end{figure}
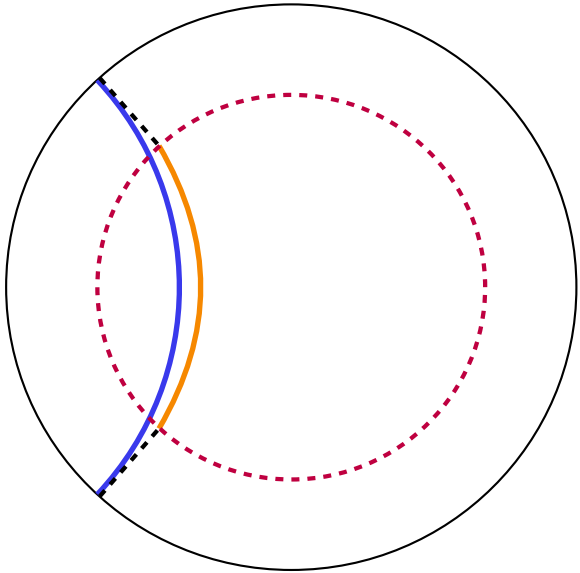

The plan of the paper is as follows. In Section \ref{sec:conformal-infinity}, we review the notion of conformal infinity and introduce a mathematical framework for studying asymptotically AdS spacetimes. In Section \ref{sec:main-sec}, we analyze the asymptotic structure of extremal surfaces in asymptotically AdS spacetimes and prove the main result of this paper: that the finite difference in area between two extremal surfaces homologous to the same boundary region is cutoff-independent. We also comment on state-dependent divergences in the entanglement entropy, and identify the class of spacetimes in which one can regulate the area of an extremal surface by subtracting off the area of a corresponding surface in the vacuum. In Section \ref{sec:discussion}, we review our results and comment on potential future applications of our techniques. Appendix \ref{app:FG} contains a careful treatment of the Fefferman-Graham prescription for asymptotic expansions of AdS metrics, including a generalization of the expansion to arbitrary matter falloff conditions. Appendix \ref{app:extremal-surfaces} contains a pedagogical introduction to the covariant analysis of extremal surfaces, including four equivalent (covariant) ways of characterizing extremality. To our knowledge, two of these four characterizations have not appeared in the literature before. The techniques developed in Appendix \ref{app:extremal-surfaces} are quite general, and may be useful more broadly in the study of holographic entanglement entropy.

We mostly use the notational conventions of \cite{WaldBook}, using a ``mostly pluses'' metric signature on spacetime and using the early Latin alphabet $a, b, \dots$ to denote ``abstract'' tensor indices. Late-alphabet Greek indices $\mu, \nu, \dots$ are reserved for expressions that only hold in a particular system of coordinates. Middle-alphabet Latin indices $j, k, \dots$ are generally used to refer to coordinate indices on submanifolds (i.e., hypersurfaces), but we make this distinction explicit whenever it appears. We work in units with $\hbar = c = 1,$ leaving Newton's constant $G_N$ explicit. All dimension-dependent expressions use the variable $d$ to denote \emph{spacetime} dimension, not the dimension of the dual CFT. We will not address the question of existence of extremal surfaces homologous to particular boundary subregions; discussions of this issue in the mathematics literature can be found in \cite{existence1, existence2, existence3}, and in the physics literature in \cite{maximin, maximin2}. We will also assume throughout that the extremal surfaces under consideration are smooth, at least in a neighborhood of the AdS boundary, and that their normal bundles are locally integrable (though this condition can be relaxed; see Appendix \ref{app:extremal-surfaces}).

\section{Conformal Infinity and Asymptotic Cutoffs}
\label{sec:conformal-infinity}

To understand radial cutoffs in AdS/CFT, we must first understand the structure of asymptotically anti-de Sitter spacetimes ``at infinity.'' The general machinery for quantifying the notion of such a spacetime boundary is that of \emph{conformal infinity}, which encodes the universal asymptotic behavior of a family of spacetime metrics that all look alike outside some bounded region. In defining conformal infinity, we follow the exposition of \cite{FG1}.

Consider first a manifold-with-boundary $\bar{M}$ with interior $M$ and boundary $\del M$. For a spacetime metric $g_{ab}$ on $M$, we will define a notion of ``infinity'' for the spacetime $(M, g_{ab})$ by using this interior metric to define an ``extended'' metric on the boundary. If the metric $g_{ab}$ already admits a smooth extension to the boundary $\del M$, as is the case when $M$ is compact, then this extension is trivial; however, it is often the case that $g_{ab}$ blows up near $\del M$ and thus admits no such smooth extension. For concreteness, one may think of vacuum $AdS_3,$ which has metric
\begin{equation}
    ds^2 = - (1 + r^2) dt^2 + (1 + r^2)^{-1} dr^2 + r^2 d \theta^2
\end{equation}
in global coordinates with the AdS radius set to $\ell_{AdS} = 1.$ This metric blows up as $r^2$ in the limit $r \rightarrow \infty.$

To get around the asymptotic blowup of the metric, we introduce the notion of a \emph{defining function} for the boundary. A defining function $z$ is a smooth map $z : \bar{M} \rightarrow \mathbb{R}$ such that
\begin{enumerate}[(i)]
    \item $z$ is positive on $M$ and zero on $\del M$,
    \item $z$ has a simple zero on $\del M$, i.e., $dz|_{\del M} \neq 0,$ and
    \item $z^k g_{ab}$ admits a smooth (nonzero) extension to the boundary for some $k > 0.$
\end{enumerate}
The defining function $z$ essentially cancels a pole of the form $1 / z^k$ in the spacetime metric; it is called a ``defining function'' because it ``defines'' the metric on the boundary. Furthermore, condition (ii) guarantees that the level sets of $z$ near $z=0$ smoothly foliate a neighborhood of the boundary. As such, level sets of $z$ can be used to regulate quantities that are formally infinite in the limit as one approaches the boundary.

Of course, there are infinitely many different defining functions that can be used to extend $g_{ab}$ to the boundary; in fact, for any defining function $z$ and smooth function $\omega$ on $\bar{M}$, $z' = e^{\omega} z$ is a defining function with the same $k$. Conversely, any two defining functions $z$ and $z'$ are related by a positive prefactor $e^{\omega} = z' / z$, whose smoothness on the boundary follows from condition (ii) above.\footnote{In particular, this duality implies that the order of the pole, $k$, is an intrinsic property of the metric and does not depend on the choice of the defining function.} It follows that for any two defining functions $z$ and $z'$, the corresponding smooth metrics $\tilde{g}_{ab} = z^k g_{ab}$ and $\tilde{g}'_{ab} = (z')^k g_{ab}$ are related by a conformal transformation
\begin{equation}
    \tilde{g}'_{ab} = e^{k \omega}\, \tilde{g}_{ab}.
\end{equation}
\emph{A priori}, no one choice of defining function is better than any other for extending the spacetime metric $g_{ab}$ ``to infinity.'' As such, the natural structure to associate with the boundary of $(M, g_{ab})$ is the \emph{conformal manifold} $(\del M, [\tilde{g}_{ab}]_{T \del M})$, defined as the boundary manifold together with the conformal equivalence class of metrics that can be obtained by choosing a defining function to extend $g_{ab}$ smoothly to the boundary. The conformal manifold $(\del M, [\tilde{g}_{ab}]_{T \del M})$ is called \emph{conformal infinity} for the spacetime $(M, g_{ab})$.

The machinery of conformal infinity gives a clean way of studying the asymptotic behavior of a spacetime. In particular, two spacetimes are said to be \emph{asymptotically equivalent} if they admit embeddings into manifolds-with-boundary with isomorphic conformal infinities. A $d$-dimensional spacetime is said to be \emph{asymptotically anti-de Sitter} if it has the same conformal infinity as vacuum anti-de Sitter spacetime in $d$ dimensions, i.e., the conformal infinity given by a manifold with topology $S_{d-2} \times \mathbb{R}$ and a conformal metric with representative
\begin{equation}
    ds^2
        = - dt^2 + d \Omega_{d-2}^2.
\end{equation}
Note that this characterizes the conformal infinity of an asymptotically \emph{global} AdS spacetime. In studying the AdS/CFT correspondence, one is also frequently interested in studying spacetimes whose conformal infinities match that of a patch of global AdS --- e.g., spacetimes whose asymptotics match that of the Poincar\'{e} patch.

Anti-de Sitter spacetime in $d$ dimensions is characterized as the universal cover of maximally symmetric spacetimes satisfying the equation\footnote{Here, as before, we implicitly use units where the AdS radius is set to $\ell_{AdS} = 1.$}
\begin{equation} \label{eq:AdS}
    R_{ab} = - (d - 1) g_{ab}.
\end{equation}
Any asymptotically anti-de Sitter spacetime must necessarily satisfy equation \eqref{eq:AdS} to leading order in $z$, i.e., to order $1/z^k$. By writing equation \eqref{eq:AdS} in coordinates adapted to the level sets of $z$ and matching the $zz$ components at leading order, it is straightforward to show that the metric pole of any asymptotically anti-de Sitter spacetime must satisfy $k=2.$ For any asymptotically AdS metric $g_{ab}$ and choice of defining function $z$, it follows that $g_{ab}$ can be written in terms of an asymptotically smooth metric as
\begin{equation}
    g_{ab} = \frac{1}{z^2} \tilde{g}_{ab}.
\end{equation}
Matching the two sides of equation \eqref{eq:AdS} at order $1/z^2$ also requires that the $zz$ component of the smooth metric is universal on the boundary, i.e., that
\begin{equation}
    \tilde{g}_{zz}|_{\del M} = 1
\end{equation}
is satisfied for any defining function $z$. By picking coordinates $x^j$ adapted to the level sets of $z$ in a neighborhood of the boundary, it follows that the metric of any asymptotically AdS spacetime may be expanded in $z$ as
\begin{equation} \label{eq:early-FG}
    ds^2 = \frac{1}{z^2} \left[ (1 + O(z))\, dz^2 + \tilde{h}_{jk} d x^j dx^k \right],
\end{equation}
where $\tilde{h}_{jk}$ is a family of induced metrics on the level sets of $z$. This expression, which constitutes a universal form for asymptotically AdS metrics, holds for any spacetime that satisfies equation \eqref{eq:AdS} at leading order in $z$, even spacetimes that do not have the same global conformal structure at infinity as vacuum $AdS_d$. Such spacetimes are called asymptotically \emph{locally} anti-de Sitter, and are conjectured to have their own holographic correspondence \cite{HM1998}.

We commented earlier in this section that level sets of $z$ near the boundary could be used to regulate quantities that are formally infinite in the limit as one approaches the boundary. In particular, we will be interested in using these level sets to regulate the areas of boundary-anchored extremal surfaces. For any choice of defining function and any small positive number $z_c$, the area of a boundary-anchored surface in the region $z > z_c$ is finite. Furthermore, the fact that the level sets of $z$ smoothly foliate a neighborhood of the boundary means that this prescription for regulating areas comes with a way of taking the limit as $z$ goes to zero.\footnote{This is exactly what is usually done in the literature when the area of a surface is regulated using some ``asymptotic radial coordinate'' --- defining functions $z$ are in one-to-one correspondence with asymptotic radial coordinates $r \sim 1/z$, and all of the usual coordinates used to regulate the areas of extremal surfaces (e.g., the global anti-de Sitter radial coordinate) have level sets that agree with the level sets of some defining function $z$.} Henceforth, we will refer to level sets of $z$ for a particular choice of defining function as ``cutoff surfaces.'' We note that a choice of defining function $z$ picks out a preferred metric on the boundary and hence breaks conformal invariance at infinity; this conformal symmetry breaking is exactly what happens in a conformal field theory when one introduces an ultraviolet cutoff.

Before proceeding to prove cutoff-covariance of the holographic entanglement entropy, we pause momentarily to note that there is something of an asymmetry between defining functions and conformal representatives of the boundary metric: a choice of defining function picks out a conformal representative on the boundary, but a choice of conformal representative on the boundary does not pick out a preferred defining function for obtaining that boundary metric. One can in principle remedy this by restricting to the class of \emph{special defining functions} \cite{GL1991, GW1999} --- those that satisfy $\tilde{g}_{zz} = 1$ not just on the boundary, but in a neighborhood of the boundary as well. It was shown in \cite{GL1991} that for any representative $\tilde{h}_{ab}$ in the conformal class of boundary metrics, there exists a \emph{unique} defining function $z$ in a neighborhood of the boundary such that $z^2 g_{ab}$ limits to $\tilde{h}_{ab}$ on the boundary and that
\begin{equation}
    \frac{1}{z^2} g^{ab} (dz)_a (dz)_b = 1
\end{equation}
is satisfied in the neighborhood where $z$ is defined. Within the class of special defining functions, then, there is a unique, preferred defining function for each conformal representative of the boundary.

Computationally, these special defining functions are quite convenient. In a system of coordinates adapted to the level sets of a special defining function, the metric takes the form
\begin{equation} \label{eq:mid-FG}
    ds^2 = \frac{1}{z^2} \left[ dz^2 + \tilde{h}_{jk} d x^j dx^k \right],
\end{equation}
which has an advantage over our earlier expression, \eqref{eq:early-FG}, in that there is no need to worry about higher-order terms in $g_{zz}$. Still, it benefits us to be precise about exactly when one should restrict their attention to the class of special defining functions. In particular, our proof of the cutoff-covariance of holographic entanglement entropy in Section \ref{sec:main-sec} holds for any defining function, not just the special defining functions. This is useful to remember because it means that in practice, one does not need to be particularly careful about which radial coordinate is used to regulate the holographic entanglement entropy --- so long as one uses a coordinate that is suitably smooth at infinity, it will correspond to a defining function and thus to a reasonable cutoff prescription. However, the special class of defining functions is necessary when comparing holographic entropies between two \emph{different} spacetimes. Extremal surfaces within a single spacetime can be regulated by a single choice of defining function; when comparing extremal surfaces in two different spacetimes, however, one must make sure to choose ``matching'' defining functions so that the regulators agree. This can be accomplished by choosing a conformal representative of the boundary metric, then using the special defining function for that representative in each of the two spacetimes. One must then of course check that the answers one obtains are independent of the choice of conformal representative. We identify the class of spacetimes for which this procedure is well-defined in our discussion of vacuum subtraction in Section \ref{sec:two-spacetimes}.

The special defining functions are used somewhat famously in the Fefferman-Graham prescription for asymptotic expansions of the metric. While the definition of an asymptotically anti-de Sitter spacetime requires only that the AdS equation \eqref{eq:AdS} is satisfied at order $1/z^2$, a stricter constraint is frequently imposed in the literature. This is the \emph{Fefferman-Graham falloff condition}, which requires that the spacetime metric satisfies\footnote{Here we employ ``little-o'' notation. Terms of order $o(z^{d-4})$ vanish \emph{strictly faster} than $z^{d-4}$ in the limit $z \rightarrow 0$.}\textsuperscript{,}\footnote{The rate of falloff in equation \eqref{eq:AdS-FG} may seem to depend on the choice of defining function $z$. However, condition (ii) above, which states that $dz$ is nonvanishing on the boundary, implies that equation \eqref{eq:AdS-FG} is $z$-independent.}
\begin{equation} \label{eq:AdS-FG}
    R_{ab} = - (d-1) g_{ab} + o(z^{d-4}).
\end{equation}
Spacetimes with this falloff have the property that for any special defining function $z$, with metric given in the form \eqref{eq:mid-FG}, the induced metric $\tilde{h}_{jk}$ on the cutoff surfaces is uniquely specified as an even power series in $z$ up to and including order $z^{d-2}$ \cite{FG1,FG2}. It is also often considered a ``physically reasonable'' falloff condition, since for spacetimes satisfying Einstein's equations, equation \eqref{eq:AdS-FG} is satisfied if and only if the bulk stress-energy tensor falls off as
\begin{equation}
    T_{ab} = o(z^{d-4})
\end{equation}
near the boundary, which is exactly the falloff one needs in order for a family of timelike observers to measure vanishing integrated energy density in a neighborhood of the boundary.\footnote{Typical matter configurations will actually fall off much faster than $o(z^{d-4})$. By solving the scalar wave equation in an asymptotically AdS background using the coordinates given in equation \eqref{eq:mid-FG}, for example, one can check that the stress-energy of normalizable Klein-Gordon fields in asymptotically anti-de Sitter spacetimes falls off near the boundary as $z^{2d-4}.$  Further comments on this observation, and on the vanishing integrated energy density measured by a family of timelike observers near the boundary, can be found in Appendix \ref{app:FG-falloff}.} However, one can easily write down asymptotically AdS metrics that do not obey this condition, and some are even believed to be dual to holographic states \cite{MW}. As we will discuss further in Section \ref{sec:two-spacetimes}, such spacetimes can have state-dependent divergences in the areas of extremal surfaces that must be treated carefully. To generalize the Fefferman-Graham expansion to arbitrary falloff conditions, we provide the following theorem, which is proved in Appendix \ref{app:FG}.

\begin{theorem}[Generalized Fefferman-Graham Expansion]
    \label{thm:FG}
    Let $z$ be a special defining function for a $d$-dimensional, asymptotically anti-de Sitter spacetime, with coordinates $x^j$ chosen along the level sets of $z$ so that the metric takes the form
    \begin{equation}
        ds^2 = \frac{1}{z^2} \left[ dz^2 + \tilde{h}_{jk} d x^j dx^k \right].
    \end{equation}
    Suppose, further, that the spacetime metric $g_{ab}$ satisfies the metric falloff condition
    \begin{equation}
        R_{ab} = - (d-1) g_{ab} + o(z^k)
    \end{equation}
    for some integer $k \geq -2.$ Then the smooth induced metric $\tilde{h}_{jk}$ is uniquely determined as a power series in $z$ up to and including order $z^{k+2}$ or $z^{d-2}$, whichever comes first. This power series is even in $z$.
\end{theorem}

\section{Asymptotic Structure of Extremal Surfaces} \label{sec:main-sec}

The HRT formula \eqref{eq:HRT} relates the entanglement entropy of a codimension-1 boundary region to the area of a codimension-2, spacelike bulk surface. A codimension-2 surface $\Sigma$ has a 2-dimensional normal bundle $\mathrm{N}(\Sigma)$. Since the normal bundle is $2$-dimensional, it admits a unique differential two-form $\mathbf{N}$ up to normalization and sign. We may uniquely specify this two-form by requiring the normalization condition\footnote{The minus sign in \eqref{eq:unit-binormal} comes from the fact that we are considering codimension-2 spacelike surfaces in Lorentz-signature spacetimes. In general, the sign of \eqref{eq:unit-binormal} should be $(-1)^s$, where $s$ is the number of minus signs in the signature of the normal bundle.}
\begin{equation} \label{eq:unit-binormal}
    N_{ab} N^{ab} = - 2
\end{equation}
and fixing the sign of $\mathbf{N}$ to match the orientation of $\Sigma.$ The two-form in the normal bundle of $\Sigma$ satisfying equation \eqref{eq:unit-binormal} is called the \emph{unit binormal} of $\Sigma$. Its normalization is chosen so that if $\epsilon_{a_1 \dots a_d}$ is the full spacetime volume form, then
\begin{equation} \label{eq:induced-volume}
    \frac{1}{2} \epsilon_{a_1 \dots b c} N^{b c}
\end{equation}
is the induced volume form on $\Sigma.$

By foliating a neighborhood of $\Sigma$ with a family of codimension-2 surfaces containing $\Sigma$, we may extend $N_{ab}$ to a binormal field in that neighborhood. In Appendix \ref{app:extremal-surfaces}, we prove that a necessary and sufficient condition for $\Sigma$ to be extremal is that any such extension\footnote{Technically, equation \eqref{eq:binormal-divergence} only holds for foliations with locally integrable normal bundles. In general, only the components of \eqref{eq:binormal-divergence} normal to $\Sigma$ are guaranteed to vanish; the vanishing of the tangent components is a consequence of normal bundle integrability. The Hodge dual of the normal components of \eqref{eq:binormal-divergence}, where the dual is taken with respect to the induced metric on the normal bundle, is proportional to the trace of the extrinsic curvature of $\Sigma$, which is well-known to vanish on extremal surfaces. This statement is proved in Appendix \ref{app:extremal-surfaces}, where we also explain how our proofs can be generalized to remove the assumption of normal bundle integrability.} satisfies
\begin{equation} \label{eq:binormal-divergence}
    \nabla_a N^{ab}|_{\Sigma} = 0.
\end{equation}
Equation \eqref{eq:binormal-divergence} might initially appear to be highly sensitive to the choice of how $\mathbf{N}$ is extended into a neighborhood of $\Sigma$. However, we show in Appendix \ref{app:extremal-surfaces} that $\Del_a N^{ab}|_{\Sigma}$ is independent of the choice of extension; in fact, it satisfies
\begin{equation} \label{eq:binormal-divergence-2}
    \Del_a N^{ab}|_{\Sigma} = Q^{a}{}_{c} \Del_a N^{cb}|_{\Sigma},
\end{equation}
where $Q_{ac}$ is the induced metric on $\Sigma.$ Since $Q_{ac}$ is everywhere tangent to $\Sigma$, the right-hand side of \eqref{eq:binormal-divergence-2} may be thought of as a directional derivative of $\mathbf{N}$ along directions tangent to $\Sigma$; it does not depend on the choice of extension. Even though we could use the expression
\begin{equation}
    Q^{a}{}_{c} \Del_a N^{cb} = 0
\end{equation}
to characterize extremality without worrying about the choice of extension, our calculations in the following subsections are made considerably simpler by using equation \eqref{eq:binormal-divergence} and leaving the (arbitrary) choice of how to extend $\mathbf{N}$ implicit.

Equation \eqref{eq:binormal-divergence} may be thought of as a covariant formulation of the extremality condition that ``both null expansions vanish,'' which appears ubiquitously in the general relativity literature and was applied to the study of holographic entanglement entropy in \cite{HRT}. To see this, note that for the two-parameter family of surfaces defined by extending $\Sigma$ along geodesics in the null normal directions $k^a$ and $\ell^a$, the unit binormal field may be written up to orientation as
\begin{equation} \label{eq:binormal-null}
    N_{ab} = \ell_a k_b - k_a \ell_b,
\end{equation}
where we have chosen to normalize the null vectors such that $\ell^a k_a = 1.$ (This follows from the uniqueness of two-forms on the normal bundle up to normalization.) Taking the divergence of \eqref{eq:binormal-null} and using the condition that the Lie derivative $(\mathscr{L}_{k} \ell)^a$ vanishes along such an extension yields the expression
\begin{equation}
    \Del_a N^{ab}|_{\Sigma} = \left[(\Del_a \ell^a) k^b - (\Del_a k^a) \ell^b\right]_{\Sigma}.
\end{equation}
The divergence of this particular extension of the unit binormal therefore vanishes on $\Sigma$ if and only if both null expansions vanish on $\Sigma.$ Since \eqref{eq:binormal-divergence} is deformation-independent, it can be checked for all extensions of $N_{ab}$ by verifying it for the extension of $N_{ab}$ along geodesics of $\ell^a$ and $k^a$. It follows that equation \eqref{eq:binormal-divergence} and the vanishing of both null expansions are equivalent.

All of the results in this section, and hence all the main results of this paper, follow from analyzing equation \eqref{eq:binormal-divergence} in coordinates adapted to the level sets of an arbitrary defining function $z$. In the following subsections, we solve equation \eqref{eq:binormal-divergence} order-by-order in $z$ to show that any two extremal surfaces homologous to the same boundary region in a given spacetime have both (i) the same (smooth) binormal $\tilde{N}^{\mu \nu} = N^{\mu \nu} / z^2$ up to and including order $z^{d-3}$, and (ii) the same \emph{coordinate} position up to and including order $z^{d-2}$, where $d$ is the spacetime dimension. From these two results, we prove both that finite linear combinations of entanglement entropies are cutoff-independent and that the ``globally minimal'' extremal surface of a boundary region is uniquely defined (except in the special case of multiple, equivalent global minima).
We then comment on state-dependent divergences in the holographic entanglement entropy, showing that extremal surfaces in spacetimes satisfying the additional constraint of the Fefferman-Graham falloff \eqref{eq:AdS-FG} have spacetime-\emph{independent} asymptotic structure. The areas of extremal surfaces in this class of spacetimes can be regulated using a vacuum subtraction protocol.

\subsection{Universal Divergences from Boundary Geometry}

\label{sec:the-hard-one}

Let $\Sigma$ be a boundary-anchored extremal surface in an asymptotically AdS spacetime, and let $z$ be an arbitrary defining function. In coordinates adapted to the level sets of $z$, the spacetime metric takes the form
\begin{equation}
    ds^2 = \frac{1}{z^2} \left[ (1 + O(z)) dz^2 + \tilde{h}_{jk} dx^j dx^k \right].
\end{equation}
In such coordinates, the components of the divergence equation \eqref{eq:binormal-divergence} for the unit binormal form of $\Sigma$ are given by
\begin{eqnarray}
    \tilde{D}_{j} \tilde{N}^{jz} \label{eq:div-jz} & = & 0, \\
    \frac{d-2}{z} \tilde{N}^{zk} - \del_z \tilde{N}^{zk} & = &  \tilde{N}^{zk} \del_{z} \ln\sqrt{|\tilde{h}|} + \tilde{D}_{j} \tilde{N}^{jk}. \label{eq:div-jk}
\end{eqnarray}
Several terms in these equations warrant further explanation. First, we have written both expressions in terms of the \emph{smooth} unit binormal $\tilde{N}^{\mu \nu} = N^{\mu \nu}/z^2$. The $1/z^2$ divergence in the metric implies that a normalized\footnote{A ``normalized'' tensor in this context is one whose contraction with its dual is $O(1)$ in $z$.} tensor with $m$ up-indices and $n$ down-indices scales near the boundary as $z^{m-n}$; it is $\tilde{N}^{\mu \nu},$ not $N^{\mu \nu}$, that admits a smooth expansion in $z$. Second, $\tilde{D}_{j}$ is the covariant derivative with respect to the smooth induced metric $\tilde{h}_{jk}$. It satisfies, e.g.,
\begin{equation}
    \tilde{D}_{j} \tilde{N}^{jk} = \del_j \tilde{N}^{jk} + \tilde{\Gamma}^{j}{}_{j \ell} \tilde{N}^{\ell k},
\end{equation}
where $\tilde{\Gamma}^{\ell}{}{j k}$ are the Christoffel symbols of the smooth induced metric $\tilde{h}_{jk}.$ Finally, although we have suppressed the notation $|_{\Sigma}$ for convenience, both equations \eqref{eq:div-jz} and \eqref{eq:div-jk} should be understood to hold only on $\Sigma.$ We will now use these equations to show that the geometry of the entangling surface $\del \Sigma$ fixes (i) the components of $\tilde{N}^{\mu \nu}$ up to and including order $z^{d-3}$, and (ii) the coordinate position of $\Sigma$ up to and including order $z^{d-2}$. These results constitute a form of universality for extremal surfaces homologous to the same boundary region in a given spacetime.

Since both equations \eqref{eq:div-jz} and \eqref{eq:div-jk} hold in \emph{any} coordinates adapted to level sets of $z$, we are free to choose coordinates along the cutoff surface at our convenience. Such a choice of coordinates does not introduce any additional coordinate-dependence to our proofs; the cutoff prescription is given by the choice of defining function $z$, and coordinates chosen along the cutoff surface do not affect the $z$-dependence of the areas of extremal surfaces. It will be especially useful for us to choose coordinates along the level sets of $z$ such that at $z=0,$ both $x_1$ and $x_2$ are non-tangent to the entangling surface $\del \Sigma$.\footnote{In general, such coordinates can only be chosen in a small neighborhood of $\del \Sigma$, which is sufficient for the purposes of our proof. In fact, it may not always be possible to choose such coordinates in a neighborhood of the full entangling surface $\del \Sigma$; in general, one can only find such coordinates in neighborhoods of the elements of an open cover of $\del \Sigma$. We will ignore this technicality, but all of our proofs can easily be generalized by performing calculations in each of these open sets and patching them together.} These coordinates can then be extended into the bulk by Lie transport along integral curves of $\Del_a z$, and for small $z$ both $x_1$ and $x_2$ will continue to have the property that they are non-tangent to the extremal surface $\Sigma.$ The advantage of such a system of coordinates is that near the boundary, $\Sigma$ can be described locally as the solution to two equations:
\begin{equation} \label{eq:sigma-functions}
    x_1 = f_1(z, x_3, \dots, x_{d-1})
    \quad \text{and} \quad
    x_2 = f_2(z, x_3, \dots, x_{d-1}).
\end{equation}
As we will see shortly, the divergence equations \eqref{eq:div-jz} and \eqref{eq:div-jk} can be used to solve for equations \eqref{eq:sigma-functions} perturbatively in $z$.

For the moment, though, let us return to equation \eqref{eq:div-jk}.\footnote{As in Appendix \ref{app:FG}, only equation \eqref{eq:div-jk} will be necessary for showing the universality of divergences among extremal surfaces that share the same entangling surface $\del \Sigma$; we will later show that equation \eqref{eq:div-jz} is automatically satisfied by the solutions of equation \eqref{eq:div-jk}, guaranteeing the perturbative \emph{existence} of extremal surfaces with particular boundary conditions.} We may expand both the unit binormal $\tilde{N}^{\mu \nu}$ and the smooth induced metric $\tilde{h}_{jk}$ perturbatively in $z$ as
\begin{eqnarray}
    \tilde{N}^{\mu \nu} & = & \tilde{N}^{\mu \nu}_{(0)} + z \tilde{N}^{\mu \nu}_{(1)} + z^2 \tilde{N}^{\mu \nu}_{(2)} + \dots, \label{eq:N-expansion} \\
    \tilde{h}_{jk} & = & \tilde{h}_{jk}^{(0)} + z \tilde{h}_{jk}^{(1)} + z^2 \tilde{h}_{jk}^{(2)} + \dots, \label{eq:h-expansion}
\end{eqnarray}
then plug these power-series expressions into equation \eqref{eq:div-jk} order-by-order to solve for $\tilde{N}^{\mu \nu}.$ At leading order in $z$, this gives the expression
\begin{equation} \label{eq:binormal-orthogonal}
    (d-2) \tilde{N}^{zk}_{(0)} = 0.
\end{equation}
For $d > 2$, this is exactly the statement that \emph{any boundary-anchored extremal surface meets the boundary orthogonally}. More generally, at order $z^{n-1}$ for $n \geq 1$, equation \eqref{eq:div-jk} is given by
\begin{equation} \label{eq:div-obyo}
    (d-2-n) \tilde{N}^{zk}_{(n)} = \frac{1}{2} \sum_{A + B + C = n} C \tilde{N}^{zk}_{(A)} \tilde{h}^{\ell m}_{(B)} \tilde{h}_{\ell m}^{(C)} + (\tilde{D}_{j} \tilde{N}^{jk})^{(n-1)},
\end{equation}
where we have used the metric determinant identity
\begin{equation}
    \del_z \ln\sqrt{|\tilde{h}|} = \frac{1}{2} \tilde{h}^{\ell m} \del_z \tilde{h}_{\ell m}.
\end{equation}
All terms on the right-hand side of \eqref{eq:div-obyo} are of order $z^{n-1}$ or lower in $\tilde{N}^{\mu \nu}$, and of order $z^{n}$ or lower in $\tilde{h}_{\ell m}.$ If we can show that these terms are determined by the equation at order $z^{n-2}$, then equation \eqref{eq:div-obyo} can be used to solve for $\tilde{N}^{zk}_{(n)}$ inductively up to and including order $n=d-3$, beyond which the coefficient on the left-hand side of \eqref{eq:div-obyo} vanishes and the induction breaks down.

There are essentially two obstructions to this inductive analysis. The first is that knowing $\tilde{h}_{\ell m}^{(n)}$ on $\Sigma$ requires knowing the coordinate position of the surface up to order $z^{n}$. While the spacetime metric $\tilde{h}_{\ell m}$ is assumed to be known at all orders in $z$ (since we have fixed a choice of spacetime), \eqref{eq:div-obyo} involves not just an arbitrary expansion in $z$ but an expansion \emph{along the extremal surface}, which deviates from its boundary coordinate position at nonzero $z$. The second obstruction is that \emph{a priori} equation \eqref{eq:div-obyo} only solves for the $zk$ component of $\tilde{N}^{\mu \nu}$ at each order, but takes the $jk$ components as input at the next order. We address these two obstructions in the following lemma.

\begin{lemma} \label{lem:surface-position}
    Let $\Sigma$ be a boundary-anchored extremal surface in an asymptotically AdS spacetime with metric $g_{ab}$ and defining function $z$. Let $x_i$ be coordinates along the level sets of $z$ such that $\Sigma$ locally satisfies two equations of the form
    \begin{equation} \label{eq:sigma-functions-lemma}
    x_1 = f_1(z, x_3, \dots, x_{d-1})
    \quad \text{and} \quad
    x_2 = f_2(z, x_3, \dots, x_{d-1}),
    \end{equation}
    i.e., such that $x_1$ and $x_2$ are non-tangent to $\Sigma$ near the boundary. Then knowledge of (i) the components of the smooth unit binormal $\tilde{N}^{\mu \nu}$ up to and including order $z^n$ and (ii) the position of the entangling surface $\del \Sigma$ suffices to specify $f_1$ and $f_2$ up to and including order $z^{n+1}.$ (In other words, the components of the unit binormal form at order $n$ determine the coordinate position of the surface at order $n+1.$)
    
    Furthermore, this knowledge of $f_1$ and $f_2$ up to and including order $z^{n+1}$ then specifies $\tilde{N}^{jk}$ at order $z^{n+1}.$
\end{lemma}

\begin{proof}
    We begin by being slightly more careful about the smooth expansions of $\tilde{N}^{\mu \nu}$ and $\tilde{h}_{jk}$ that appear in equations \eqref{eq:N-expansion} and \eqref{eq:h-expansion}. $\tilde{N}^{\mu \nu}$ is only defined on $\Sigma$, since its value off of $\Sigma$ depends on a choice of deformation, so its expansion in $z$ along $\Sigma$ is straightforward. The induced metric $\tilde{h}_{jk}$, however, is known in all of spacetime by assumption. We must be careful in what we mean when we say we ``expand $\tilde{h}_{jk}$ along $\Sigma$.''
    
    For the purposes of this lemma, we distinguish between
    \begin{equation}
        \tilde{h}_{jk}^{\,\Sigma}(z; x_3, \dots, x_{d-1}),
    \end{equation}
    which denotes the components of $\tilde{h}_{jk}$ as a function of $\Sigma$, and the full spacetime metric
    \begin{equation}
        \tilde{h}_{jk}(x_1, x_2, z; x_3, \dots, x_{d-1}),
    \end{equation}
    which is a function of all $d$ spacetime coordinates. Since $x_1$ and $x_2$ on $\Sigma$ are locally expressible as functions of the other coordinates, these two expressions are related by
    \begin{equation}
        \tilde{h}_{jk}^{\,\Sigma}(z; x_3 \dots x_{d-1})
        =
        \tilde{h}_{jk}(f_1(z, x_3, \dots), f_2(z, x_3, \dots), z; x_3, \dots, x_{d-1}).
    \end{equation}
    Differentiating with respect to $z$ gives the expression
    \begin{equation} \label{eq:precise-h-expansion}
        \frac{\del \tilde{h}_{jk}^{\,\Sigma}}{\del z}
        =
        \frac{\del \tilde{h}_{jk}}{\del x_1} \frac{\del f_1}{\del z} + \frac{\del \tilde{h}_{jk}}{\del x_2} \frac{\del f_2}{\del z} + \frac{\del \tilde{h}_{jk}}{\del z}.
    \end{equation}
    Since $\tilde{h}_{jk}$ is known to all orders in $z$, we conclude that knowing $\tilde{h}_{jk}^{\, \Sigma}$ to order $z^n$ requires knowing $f_1$ and $f_2$ to the same order. A similar statement was made in the paragraph preceding the statement of this lemma; equation \eqref{eq:precise-h-expansion} makes that statement precise.
    
    We now proceed to prove the first claim of the lemma, that knowing $\tilde{N}^{\mu \nu}$ up to order $z^n$ specifies $f_1$ and $f_2$ up to order $n+1$. Since $\Sigma$ is a level set of both $x_1 - f_1$ and $x_2 - f_2$, its normal bundle has two preferred one-forms given by the gradients of these expressions:
    \begin{eqnarray}
        v_a & = & \Del_a (x_1 - f_1), \label{eq:v-vector} \\
        w_a & = & \Del_a (x_2 - f_2). \label{eq:w-vector}
    \end{eqnarray}
    Since there is only one two-form on the normal bundle up to orientation and normalization, the unit binormal $\textbf{N}$ must satisfy
    \begin{equation}
        N_{ab} \propto v_a w_b - v_b w_a.
    \end{equation}
    Fixing the constant of proportionality with equation \eqref{eq:unit-binormal} gives the expression\footnote{For convenience, we have oriented $\Sigma$ so that the sign in front of this expression is positive.}
    \begin{equation} \label{eq:binormal-vw}
        N_{ab} = \frac{1}{\sqrt{v^2 w^2 - (v \cdot w)^2}} (v_a w_b - v_b w_a),
    \end{equation}
    where the dot product appearing in the denominator is understood to be the inner product with respect to the metric. Using equations \eqref{eq:v-vector} and \eqref{eq:w-vector}, one can check that the $z1$ and $z2$ components of \textbf{N} satisfy
    \begin{equation}
        N_{z1} = \frac{\del f_2}{\del z} N_{12}
        \quad \text{and} \quad
        N_{z2} = - \frac{\del f_1}{\del z} N_{12}.
    \end{equation}
    Just as we replaced the up-index binormal with its smooth counterpart $\tilde{N}^{\mu \nu} = N^{\mu \nu} / z^2,$ we may replace the down-index binormal with its smooth counterpart $\bar{N}_{\mu \nu} = z^2 N_{\mu \nu}.$ Performing this substitution and solving for the derivatives of $f_1$ and $f_2$ yields the expressions
    \begin{equation}
        \frac{\del f_1}{\del z} = - \frac{\bar{N}_{z2}}{\bar{N}_{12}}
        \quad \text{and} \quad
        \frac{\del f_2}{\del z} = \frac{\bar{N}_{z1}}{\bar{N}_{12}}.
    \end{equation}
    
    We conclude that if $\bar{N}_{\mu \nu}$ is known to order $z^{m}$, then $f_1$ and $f_2$ are known at order $z^{m+1}.$ We may write $\bar{N}_{\mu \nu}$ in terms of $\tilde{N}^{\mu \nu}$ as
    \begin{equation}
        \bar{N}_{\mu \nu} = \tilde{g}_{\mu \lambda}^{\,\Sigma}\, \tilde{g}_{\nu \rho}^{\,\Sigma}\, \tilde{N}^{\lambda \rho}.
    \end{equation}
    It follows that knowledge of $\tilde{N}^{\mu \nu}$ and $\tilde{g}_{\mu \nu}^{\, \Sigma}$ up to order $z^{m}$ specifies $f_1$ and $f_2$ at order $z^{m+1}.$ We argued above, however, that $\tilde{g}_{\mu \nu}^{\, \Sigma}$ at order $z^m$ is specified by knowledge of $f_1$ and $f_2$ at order $z^{m}.$ By specifying $f_1$ and $f_2$ at order $z^0$ as a base case, which we have done by specifying the boundary entangling surface $\del \Sigma$, we may therefore solve for $f_1$ and $f_2$ inductively up to order $z^{n+1}$ by using knowledge of $\tilde{N}^{\mu \nu}$ up to order $z^n,$ as desired.
    
    The second claim of the lemma is now quite straightforward. From equation \eqref{eq:binormal-vw}, we see that the $jk$ components of $N_{\mu \nu}$ satisfy
    \begin{equation}
        N_{jk} = \frac{1}{\sqrt{v^2 w^2 - (v \cdot w)}} (v_j w_k - v_k w_j).
    \end{equation}
    By replacing the inner product in the denominator with an inner product with respect to the \emph{smooth} metric $\tilde{g}_{\mu \nu},$ one can pull out a factor of $z^2$ and obtain the expression
    \begin{equation} \label{eq:smooth-binormal-cutoff}
        \bar{N}_{jk} = \frac{1}{\sqrt{(\tilde{g}_{ab} v^a v^b) (\tilde{g}_{cd} w^c w^d) - (\tilde{g}_{ab} v^a w^b)^2}} (v_j w_k - v_k w_j).
    \end{equation}
    By expanding this equation in powers of $z$ and comparing with the explicit expressions for $v_a$ and $w_a$ given in equations \eqref{eq:v-vector} and \eqref{eq:w-vector}, we see that every term on the right-hand side of equation \eqref{eq:smooth-binormal-cutoff} is determined by knowledge of $f_1$ and $f_2$ up to the order of the left-hand side. It follows that $\bar{N}_{jk}^{(n+1)}$ is determined by $f_1$ and $f_2$ at order $z^n.$ A similar argument to that given in the previous paragraph shows that the smooth raised-index binormal, $\tilde{N}^{jk}_{(n+1)}$, is determined by the same data.
    
    Note that the $zk$ component of $\tilde{N}^{\mu \nu}$ at a given order is \emph{not} determined by knowledge of $f_1$ and $f_2$ at the same order. This is why the divergence equation, \eqref{eq:div-obyo}, is needed to solve for the unit binormal of an extremal surface order-by-order.
\end{proof}

\vspace{0.4cm}

With this lemma in hand, we may now return to equation \eqref{eq:div-obyo} and solve for the unit binormal perturbatively in $z$. As a base case, we showed in equation \eqref{eq:binormal-orthogonal} that $\tilde{N}^{zk}_{(0)}$ vanishes for any boundary-anchored extremal surface. The other components of $\tilde{N}^{\mu \nu}$ at this order, $\tilde{N}^{jk}_{(0)},$ are determined by the geometry of the entangling surface $\del \Sigma.$ At higher orders, we can use equation \eqref{eq:div-obyo} to solve for $\tilde{N}^{zk}_{(n)}$ as a function of $\tilde{N}^{\mu \nu}_{(p)}$ for $p < n$ and $\tilde{h}_{jk}^{(q)}$ for $q \leq n.$ Using Lemma \ref{lem:surface-position} and equation \eqref{eq:div-obyo} in tandem, we can solve for these quantities in terms of determined lower-order quantities.\footnote{One subtlety is that the term in \eqref{eq:div-obyo} involving a covariant derivative, $\tilde{D}_{j} \tilde{N}^{jk},$ might seem to depend on the value of $\tilde{N}^{jk}$ \emph{off} of the extremal surface $\Sigma$, which is not specified by our inductive procedure. Corollary \ref{cor:corollary} in Appendix \ref{app:extremal-surfaces}, however, shows that the covariant divergence of $\mathbf{\tilde{N}}$ depends only on the value of $\mathbf{\tilde{N}}$ on $\Sigma.$} It follows that boundary data alone suffices to specify the smooth unit binormal $\tilde{N}^{\mu \nu}$ up to and including order $z^{d-3},$ after which the inductive procedure breaks down due to the vanishing of the left-hand side of \eqref{eq:div-obyo}. In particular, this implies that \emph{any two extremal surfaces anchored to the same boundary in the same spacetime have the same unit binormal up to and including order $z^{d-3}$}. We also note that taking the divergence of equation \eqref{eq:div-jk} with respect to $\tilde{D}_{k}$ yields the expression
\begin{equation}
    \frac{d-2}{z} \tilde{D}_{k} \tilde{N}^{zk} - \del_z \tilde{N}^{zk} = \tilde{\Gamma}^{\ell}{}_{\ell z} \tilde{D}_{k} \tilde{N}^{zk}.
\end{equation}
Solving this equation order-by-order with the inductive assumption that $(\tilde{D}_{k} \tilde{N}^{zk})_{(q)}$ vanishes for $q < n$, which is satisfied by the base case $q=0$ according to equation \eqref{eq:binormal-orthogonal}, shows that $\tilde{D}_{k} \tilde{N}^{zk}$ vanishes up to and including order $z^{d-3}$ for any surface that satisfies equation \eqref{eq:div-jk}. In other words, solutions to the the cutoff-surface components of the divergence equation, \eqref{eq:div-jk}, automatically satisfy the $z$ component, \eqref{eq:div-jz}, up to the same order. It follows that extremal surfaces anchored to a particular boundary region are not only unique up to order $z^{d-2},$ but exist locally up to this order as well.

We have proven what we set out to prove: that the form of $\tilde{N}^{\mu \nu}$ for an extremal surface with given boundary $\del \Sigma$ is universal up to the addition of terms at order $o(z^{d-3}).$ We stress that by \emph{universal}, we mean that the binormal is universal among extremal surfaces homologous to the same boundary region within a \emph{single} spacetime; we will comment on universality among extremal surfaces in \emph{different} spacetimes in subsection \ref{sec:two-spacetimes}. To see what this universal behavior of the unit binormal means for the area divergences of extremal surfaces, we write the intersection of $\Sigma$ with a level set of $z$ as $\sigma_z.$ $\sigma_z$ is a codimension-3 surface in the full spacetime, and a codimension-2 surface within the level set of $z$. The area of $\sigma_z$ is given by\footnote{Recall that the induced volume form on $\Sigma$ is given by equation \eqref{eq:induced-volume}.}
\begin{equation}
    \text{Area}(\sigma_z) = \frac{1}{2} \int_{\sigma_z} \epsilon_{\dots abc} z^a N^{bc}.
\end{equation}
By replacing each tensor in this expression with its smooth counterpart, we obtain the expression
\begin{equation} \label{eq:area-divergence}
    \text{Area}(\sigma_z) = \frac{1}{2 z^{d-2}} \int_{\sigma_z} \tilde{\epsilon}_{\dots abc} z^a \tilde{N}^{bc}.
\end{equation}
Since $\tilde{N}^{bc}$ is determined to order $z^{d-3}$ by boundary data, it follows that the divergent terms in \eqref{eq:area-divergence} are universal among extremal surfaces with the same boundary anchor. The divergent piece of the area of $\Sigma$ is obtained by integrating \eqref{eq:area-divergence} in $z$ near $z=0$; if the divergent terms in \eqref{eq:area-divergence} are determined by boundary data, then any two extremal surfaces with the same boundary anchor have the same area divergences with respect to the defining function $z$. In particular, this implies that the \emph{difference} in their area is finite in the cutoff prescription provided by the choice of $z$.

Thus far, all of our analysis has involved picking a particular defining function $z$ and showing that boundary divergences of homologous extremal surfaces are universal with respect to $z$. To show that the holographic entanglement entropy is cutoff-covariant, we must verify that the finite difference in area between two homologous extremal surfaces is independent of the choice of defining function. This is the subject of the following subsection.

\subsection{Cutoff-Independence of Finite Differences in Holographic Entropy}
\label{sec:the-big-one}

Given a choice of defining function $z$, the area of an extremal surface $\Sigma$ can be split into a ``finite piece'' and a ``divergent piece'' by introducing a sharp cutoff at some fixed level set $z = z_c$. The area of the extremal surface can then be written as
\begin{equation}
    \text{Area}(\Sigma) = \lim_{z_c \rightarrow 0} \frac{1}{2} \int_{\Sigma > z_c} \epsilon_{....ab} N^{ab},
\end{equation}
where the integral is taken over the portion of $\Sigma$ lying in the compact region $z > z_c.$ By evaluating this integral at finite $z_c$, the area can be written as a Laurent series in $z_c$ with the addition of logarithmic terms. Terms proportional to $z_c^{k}$ for $k>0$ vanish in the limit $z_c \rightarrow 0$, the term proportional to $z_c^{0}$ is called the ``finite piece'' of the area, and the remaining terms are collectively called the ``divergent piece'' of the area. Since we showed in the previous subsection that any two extremal surfaces anchored to the same boundary region have identical divergences, it follows that the difference in their areas under this cutoff prescription is simply the difference in these finite pieces.

Showing that holographic entanglement entropy is cutoff-covariant amounts to showing that this \emph{finite} area difference is independent of cutoff prescription. To show this, we will ask how the finite piece of the area of $\Sigma$ changes when a different defining function is chosen to implement the cutoff. We claim that any two extremal surfaces with the same boundary anchor experience the \emph{same} finite change in area under a change in cutoff prescription; this implies that the finite difference of their areas is cutoff-independent.

Recall from Section \ref{sec:conformal-infinity} that any two defining functions are related by a conformal transformation
\begin{equation}
    z = e^{\omega} z'.
\end{equation}
Changing our cutoff prescription from $z$ to $z'$ therefore amounts to moving our cutoff surface from $z = z_c$ to
\begin{equation} \label{eq:cutoff-transform}
    z = e^{\omega} z_c.
\end{equation}
Note that $\omega$ generally depends on $z$, so equation \eqref{eq:cutoff-transform} should be understood not as an explicit expression for $z$ but as an equation that must be solved to find the position of the new cutoff surface in coordinates adapted to the original cutoff $z$. That this equation admits a unique solution is guaranteed by the fact that the level sets of any defining function smoothly foliate a neighborhood of the boundary (cf. Section \ref{sec:conformal-infinity}).

The change in the regulated area of $\Sigma$ induced by changing the cutoff prescription is simply the area of $\Sigma$ contained between surfaces $z = z_c$ and $z = e^{\omega} z_c.$ As in Lemma \ref{lem:surface-position}, we may choose coordinates $x_1$ and $x_2$ that are non-tangent to $\Sigma$ near the boundary, so an integral over $\Sigma$ may be treated as an integral over the coordinates $z, x_3, \dots, x_{d_1}$ with $x_1$ and $x_2$ fixed to satisfy equations \eqref{eq:sigma-functions}. In these coordinates, the change in the regulated area of $\Sigma$ under a change of cutoff prescription may be written as
\begin{equation}
    \frac{1}{2} \int dx_3 \dots dx_{d-1}\, \int_{z = e^{\omega} z_c}^{z_c} dz\, \epsilon_{3 4 \dots (d-1) z 1 2} N^{12}.
\end{equation}
Terms in the integrand that are finite or zero in the limit $z \rightarrow 0$ will contribute to this expression at order $z_c$ or higher. It follows that the change in the finite piece of the area of $\Sigma$ induced by changing the cutoff prescription is determined by $z$-divergent terms in the integrand. In the previous subsection, however, we ordered that the $z$-divergent terms in this integrand are determined by boundary data alone. It follows that any two extremal surfaces with a common boundary anchor $\del \Sigma$ experience the same finite change in area under a change of cutoff. We conclude, as claimed, that the finite difference in the area of two such surfaces is cutoff-independent.

This result has important implications for the AdS/CFT correspondence. For one, it implies that finite sums of entanglement entropies such as the mutual information of non-adjacent regions,
\begin{equation} \label{eq:MI2}
    I(A:B) = S(A) + S(B) - S(A\cup B),
\end{equation}
are cutoff-independent under the holographic dictionary \eqref{eq:HRT}, since they can be expressed as area differences of extremal surfaces that share a common boundary. Similarly, it implies that the notion of a ``globally minimal surface'' anchored to a particular boundary region is well-defined. One simply considers all extremal surfaces anchored to the boundary region, then picks out the one whose finite area difference with all other candidates is negative. Unless the area difference between two candidates is zero and the global minimum is degenerate, the globally minimal surface is uniquely prescribed. This is an essential component of the holographic dictionary, as being able to identify the globally minimal surface of a particular boundary region is essential both in locating the entanglement wedge  \cite{EW} and in determining the values of finite information-theoretic quantities; for example, determining the mutual information \eqref{eq:MI2} requires knowing which of two candidate extremal surfaces for $A\cup B$ is globally minimal.

A similar argument used to prove the cutoff-independence of finite area differences can be used to prove a corollary mentioned in the introduction of this paper: that two common prescriptions for cutting off the areas of extremal surfaces are equivalent. As mentioned in the introduction, there are two prescriptions in the literature for cutting off the area of an extremal surface. The first, which is the one we have used throughout this paper, is to impose a cutoff $z=z_c$ directly on the boundary-anchored extremal surface $\Sigma.$ The second is to transport the entangling surface $\del \Sigma$ to the cutoff surface $z = z_c$ along integral curves of $z$, then to find a new extremal surface with this transported entangling surface as its anchor. These two prescriptions generally produce different surfaces with different areas; cf. Figure \ref{fig:two-cutoff-prescriptions-1} for an example in vacuum $AdS_3$.

One might initially hope that these two prescriptions are equivalent at the level of a single surface, i.e., that since the ``boundary-anchored'' and ``cutoff-anchored'' surfaces converge in the limit $z_c \rightarrow 0$, their finite pieces are the same. However, this is not the case. The coordinate distance between two such surfaces scales as $z_c^2$, since this is the order at which a boundary-anchored extremal surface deviates from being orthogonal to the boundary. From the previous subsection, though, we know that the area of the codimension-$3$ surface obtained by intersecting an extremal surface with the level set $z = z_c$ scales as $z_c^{3-d}$. These observations collectively imply that the difference in area between the two prescriptions for a single surface scales as $z_c^{5-d}$ --- it vanishes for $d < 5,$ but is nonvanishing (and even divergent!) in higher dimensions.

In analogy with our previous arguments, however, the equivalence between these two prescriptions is saved by the fact that they give equivalent answers for the area \emph{difference} between two extremal surfaces with a common boundary. Before proving this, it will first be useful to understand how one should identify boundary-anchored extremal surfaces with $z_c$-anchored extremal surfaces. Consider, for example, the setup shown in Figure \ref{fig:two-cutoff-prescriptions-2} for a constant-time slice of vacuum $AdS_3$. There are two candidate extremal surfaces for the boundary region $A\cup B$ --- one that is simply the union of the extremal surfaces for $A$ and $B$ individually, and one that connects their boundaries. If one transports $A$ and $B$ to the cutoff surface along integral curves of $z$, there are again two candidate extremal surfaces for the union $A \cup B$. In this case, it is intuitively obvious how one should identify surfaces between the two prescriptions; the ``disconnected'' boundary-anchored surface maps to the ``disconnected'' cutoff-anchored surface, and the connected one maps to the connected one. In general, though, one must be precise about how these surfaces are identified.

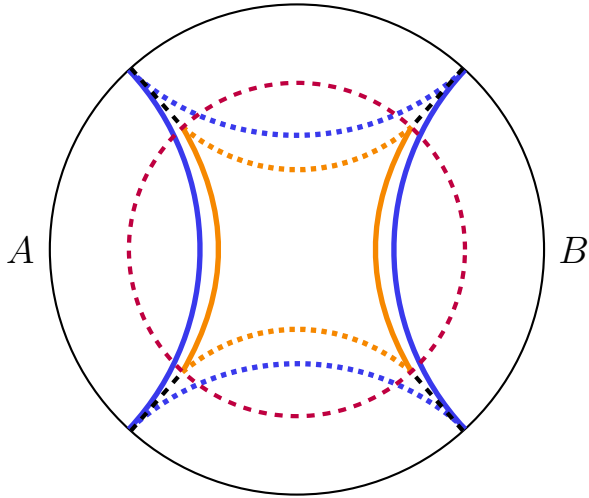
\begin{figure}
    \centering
	\begin{tikzpicture}[scale=1.3, thick]
	\draw (0, 0) circle (2.5cm);

	\draw [line width=2pt, lightgray!30!blue] (-1.7, 1.83) arc (42.89:-42.89:2.69);
	\draw [line width=2pt, lightgray!30!blue] (1.7, 1.83) arc (137.11:222.89:2.69);
	
	\draw [line width=2pt, lightgray!30!blue, dotted] (-1.7, 1.83) arc (227.11:312.89:2.49);
	\draw [line width=2pt, lightgray!30!blue, dotted] (1.7, -1.83) arc (47.11:132.89:2.49);
    
    \draw [line width=2pt, black!15!lime!10!orange] (-1.16, 1.25) to [out=-60, in=60] (-1.16,-1.25);
    \draw [line width=2pt, black!15!lime!10!orange] (1.16, 1.25) to [out=240, in=120] (1.16,-1.25);
    
    \draw [line width=2pt, black!15!lime!10!orange, dotted] (-1.16, 1.25) to [out=-40, in=220] (1.16,1.25);
    \draw [line width=2pt, black!15!lime!10!orange, dotted] (-1.16, -1.25) to [out=40, in=140] (1.16,-1.25);
    
    \draw [ultra thick, dashed] (-1.68, 1.85) to (-1.16, 1.25);
    \draw [ultra thick, dashed] (-1.68, -1.85) to (-1.16, -1.25);
    \draw [ultra thick, dashed] (1.68, 1.85) to (1.16, 1.25);
    \draw [ultra thick, dashed] (1.68, -1.85) to (1.16, -1.25);
    
	\draw [purple, ultra thick, dashed] (0,0) circle (1.7cm);
	
	\node at (-2.8, 0) (A) {\Large $A$};
	\node at (2.8, 0) (B) {\Large $B$};
	
	\end{tikzpicture}
	
	\caption{Two different candidate extremal surfaces for boundary region $A \cup B$ in vacuum $AdS_3,$ both in the ``boundary-anchored'' cutoff prescription and the ``cutoff-anchored'' cutoff prescription. The ``disconnected'' candidates are shown as solid curves, while the ``connected'' candidates are shown as dotted curves. In this instance, it is clear that the ``dotted'' surfaces and ``solid'' surfaces should be identified with one another; to make this precise in general, one must use the language of extremal foliations (cf. Fig \ref{fig:two-cutoff-prescriptions-3}).}
	\label{fig:two-cutoff-prescriptions-2}
\end{figure}

The prescription we propose is as follows. Let $R$ be a region on the boundary of an asymptotically AdS spacetime, and let $R_{z_c}$ be the corresponding region on the cutoff surface obtained by transporting $R$ along integral curves of $z$. Assuming sufficient stability in the extremal surfaces anchored to $R$, the class of extremal surfaces anchored to $R$ should be isomorphic to the class of extremal surfaces anchored to $R_{z_c}$, at least when $z_c$ is small. Roughly speaking, if there exists a stable extremal surface anchored to $R$, then the change in boundary conditions induced by transporting $R$ to $R_{z_c}$ is so small in the limit of small cutoff that there should exist a corresponding solution anchored to $R_{z_c}$. Likewise, a small cutoff should not introduce any new extremal surfaces that do not exist at the boundary. That said, we do not address questions of existence in this paper, and there may be subtleties here that require further investigation. For the moment, we assume that the set of cutoff-anchored extremal surfaces is in bijection with the set of boundary-anchored extremal surfaces.

What, then, is the nature of this bijection? To guide us, think back on the example of the ``disconnected'' and ``connected'' candidate extremal surfaces for a boundary region $A \cup B$ sketched in Figure \ref{fig:two-cutoff-prescriptions-2}. It is clear, intuitively, that the ``dotted'' boundary-anchored surface should map to the dotted cutoff-anchored surface, and likewise that the ``solid'' boundary-anchored surface should map to the solid cutoff-anchored surface. We also note that it is possible to smoothly foliate between the dotted surfaces within the cutoff region, and likewise between the solid surfaces, but not from one to the other. Furthermore, in vacuum $AdS_3,$ these foliations can be constructed such that each leaf is itself an extremal surface. This leads us to conjecture the following (cf. Figure \ref{fig:two-cutoff-prescriptions-3}).
\begin{conjecture}
    Let $R$ be a spacelike, codimension-$1$ subregion of the boundary in an asymptotically AdS spacetime, with $R_{z_c}$ the corresponding ``transported'' region on the cutoff surface. Let $\Sigma$ be an extremal surface homologous to $R$, with $\Gamma_1$ the portion of $\Sigma$ contained in the region $z > z_c$. Then there exists a unique extremal surface $\Gamma_2$ anchored to $R_{z_c}$ such that there exists a foliation of extremal surfaces interpolating between $\Gamma_1$ and $\Gamma_2$, at least when $z_c$ is small. This is the cutoff-anchored extremal surface one should associate to $\Sigma.$
\end{conjecture}
It would be interesting to address this conjecture more carefully in terms of the existence of extremal surfaces with particular boundary conditions, and to determine the precise stability conditions that must be imposed on $\Sigma$ for it to hold. For the moment, however, we assume it is generally true. A foliation of this type is sketched for vacuum $AdS_3$ in Figure \ref{fig:two-cutoff-prescriptions-3}.

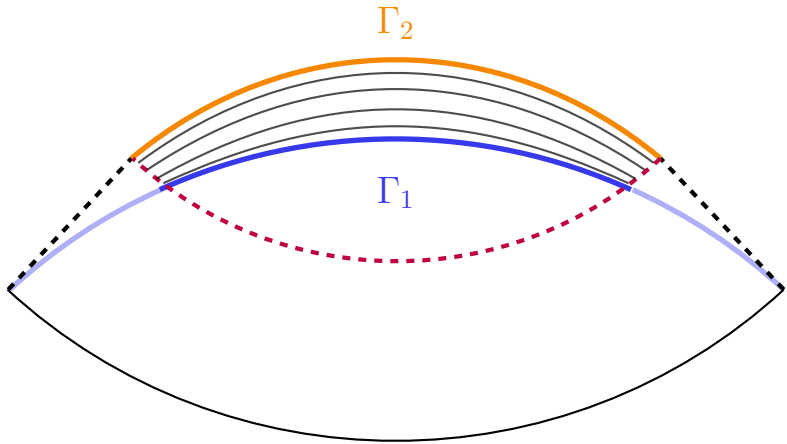
\begin{figure}
    \centering
	\begin{tikzpicture}[scale=3, thick]
	\draw (1.7, -1.83) arc (-47.11:-132.89:2.5);

	\draw [line width=2pt, lightgray!30!blue, opacity=0.4] (1.7, -1.83) arc (47.11:65.5:2.49);
	\draw [line width=2pt, lightgray!30!blue, opacity=1] (1.03, -1.388) arc (65.5:114.5:2.49);
	\draw [line width=2pt, lightgray!30!blue, opacity=0.4] (-1.03, -1.388) arc (114.5:132.89:2.49);
    
    \draw [line width=2pt, black!15!lime!10!orange] (-1.16, -1.25) to [out=40, in=140] (1.16,-1.25);
    
    \draw [ultra thick, dashed] (-1.7, -1.83) to (-1.16, -1.25);
    \draw [ultra thick, dashed] (1.7, -1.83) to (1.16, -1.25);
    
	\draw [purple, ultra thick, dashed] (1.16, -1.25) arc (-47.11:-132.89:1.7cm);
	
	\node [lightgray!30!blue] at (0, -1.4) (Gamma1) {\Large $\Gamma_1$};
	\node [black!15!lime!10!orange] at (0, -0.65) (Gamma2) {\Large $\Gamma_2$};
	\node at (0,-0.2) {};
	
	\draw [opacity=0.7] (-1.13, -1.27) to [out=37, in=143] (1.13,-1.27);
	\draw [opacity=0.7] (-1.09, -1.3) to [out=34, in=146] (1.09,-1.3);
	\draw [opacity=0.7] (-1.05, -1.34) to [out=30, in=150] (1.05, -1.34);
	\draw [opacity=0.7] (-1.02, -1.36) to [out=25, in=155] (1.02, -1.36);
	
	\end{tikzpicture}
	
	\caption{A boundary-anchored extremal surface $\Gamma_1$ and the corresponding cutoff-anchored extremal surface $\Gamma_2$, drawn here in a portion of a constant-time slice of vacuum $AdS_3$. In order to identify these surfaces, we assume the existence of a smooth foliation of extremal surfaces, sketched here as gray curves, that interpolates between the two. Using Stokes' theorem, we can compute the difference in area between $\Gamma_1$ and $\Gamma_2$ by performing an integral along the intersection of this foliation with the cutoff surface $z = z_c$. (Cf. equation \eqref{eq:stokes}.)}
	\label{fig:two-cutoff-prescriptions-3}
\end{figure}

This prescription is more than just a way of mapping boundary-anchored surfaces to cutoff-anchored surfaces; it also provides a mathematically convenient way of comparing the areas of surfaces between the two classes. If $\Gamma_1$ is the portion of the boundary-anchored surface $\Sigma$ inside $z = z_c$, and $\Gamma_2$ is the associated cutoff-anchored extremal surface, then the area difference between the two is given by
\begin{equation}
    \Delta \text{Area}
        = \int_{\Gamma_2} \bm{\alpha} - \int_{\Gamma_1} \bm{\alpha},
\end{equation}
where $\bm{\alpha} = \ast\, \mathbf{N}$ is the induced volume form. The existence of a smooth foliation of extremal surfaces between $\Gamma_1$ and $\Gamma_2$, as shown in Figure \ref{fig:two-cutoff-prescriptions-3}, implies the existence of a smooth field of induced volume forms $\bm{\alpha}$ in a region that interpolates between the two surfaces. Using Stokes' theorem, we may then write the change in area as
\begin{equation} \label{eq:area-change-penultimate}
    \Delta \text{Area}
        = \int_{B} \dd \bm{\alpha} - \int_{S} \bm{\alpha},
\end{equation}
where $B$ is the bulk region foliated by the family of extremal surfaces that interpolate between $\Gamma_1$ and $\Gamma_2,$ and $S$ is the intersection of this foliation with the cutoff surface $z = z_c.$ We show in Appendix \ref{app:extremal-surfaces} that the exterior derivative of the induced volume form on an extremal surface vanishes. This implies that the first term of \eqref{eq:area-change-penultimate} vanishes, and hence that the change in area between $\Gamma_1$ and $\Gamma_2$ is given by
\begin{equation} \label{eq:stokes}
    \Delta \text{Area}
        = - \int_{S} \bm{\alpha} = - \frac{1}{2} \int_{S} \epsilon_{\dots ab} N^{ab}.
\end{equation}
The fact that the coordinate location of $\Sigma$ is specified to order $z^{d-2}$ by boundary data implies that the coordinate location of $S$ can be specified to the same order. This fact, coupled with the fact that all divergent terms in the integrand of \eqref{eq:stokes} are determined by the entangling surface $\del \Sigma,$ implies that the finite part of the area difference between $\Gamma_1$ and $\Gamma_2$ is determined entirely by the entangling surface $\del \Sigma.$ It follows that any two extremal surfaces with the same boundary anchor undergo the same finite change in area when mapping between the ``boundary-anchored'' and ``cutoff-anchored'' cutoff prescriptions. We conclude that these two prescriptions give equivalent answers for the finite area difference between homologous extremal surfaces, as desired.

\subsection{Vacuum Subtraction and State-Dependent Divergences}
\label{sec:two-spacetimes}

Thus far, we have considered the cutoff-covariance of finite area differences between extremal surfaces with a common boundary anchor in a \emph{single} spacetime. However, in Section \ref{sec:conformal-infinity} and Appendix \ref{app:FG}, we showed that asymptotically AdS spacetimes satisfying the Fefferman-Graham falloff condition
\begin{equation} \label{eq:FG-sec3}
    R_{ab} = - (d-1) g_{ab} + o(z^{d-4})
\end{equation}
have universal asymptotic structure with respect to \emph{special} defining functions --- defining functions $z$ that put the metric in the form
\begin{equation}
    ds^2 = \frac{1}{z^2} \left[ dz^2 + \tilde{h}_{jk} dx^{j} dx^{k} \right].
\end{equation}
More specifically, in Theorem \ref{thm:FG} we claimed that in spacetimes satisfying equation \eqref{eq:FG-sec3}, the smooth induced metric $\tilde{h}_{jk}$ is determined up to and including order $z^{d-2}$ by specifying the boundary representative $\tilde{h}_{jk}^{(0)}.$ In subsection \ref{sec:the-hard-one}, we showed that the components of the smooth binormal $\tilde{N}^{\mu \nu}$ are determined up to and including order $z^{d-2}$ by the entangling surface $\del \Sigma$ and components of the metric up to and including order $z^{d-3}.$ If two spacetimes have the same components of $\tilde{h}_{jk}$ to order $z^{d-3}$, then it follows that extremal surfaces that share a common boundary region have universal divergences in \emph{both} spacetimes. We say that such spacetimes have state-\emph{independent} divergences in the holographic entanglement entropy.

Note that this condition is actually weaker than the Fefferman-Graham falloff condition; to specify the metric to order $z^{d-3},$ we need only impose equation \eqref{eq:FG-sec3} up to the addition of terms at order $o(z^{d-5}).$ In particular, in $d=3$, \emph{any} asymptotically AdS spacetime has state-independent divergences in the holographic entanglement entropy. This statement is made precise in the following theorem. The proof follows immediately from arguments given in subsection \ref{sec:the-hard-one}.

\begin{theorem} \label{thm:state-dependence}
    Let $\mathcal{M}_{1}$ and $\mathcal{M}_{2}$ be two asymptotically anti-de Sitter spacetimes of spacetime dimension $d$, both of which satisfy
    \begin{equation} \label{eq:relaxed-FG}
        R_{ab} = - (d-1) g_{ab} + o(z^{d-5}).
    \end{equation}
    Fix a conformal representative $\tilde{h}_{jk}^{(0)}$ of the boundary metric, and let $z$ be its special defining function. If $\Sigma_1$ and $\Sigma_2$ are extremal surfaces in the two different spacetimes that have the same boundary anchor at conformal infinity, then they have
    \begin{enumerate}[(i)]
        \item the same coordinate position in coordinates adapted to level sets of $z$ up to and including order $z^{d-2}$, and
        \item the same smooth unit binormal $\tilde{N}^{\mu \nu}$ up to and including order $z^{d-3}.$
    \end{enumerate}
\end{theorem}

An identical proof to that given in subsection \ref{sec:the-big-one} to show cutoff-invariance of finite area differences within a given spacetime shows that the difference in area between $\Sigma_1$ and $\Sigma_2$ is both (i) finite and (ii) independent of the choice of special defining function $z$. In particular, this implies that the area of an extremal surface in an arbitrary Fefferman-Graham spacetime --- or, in fact, any spacetime satisfying the weaker falloff condition given by \eqref{eq:relaxed-FG} --- can be regulated in a cutoff-independent manner by subtracting off the area of a corresponding surface in vacuum AdS. This is exactly the ``vacuum subtraction'' protocol that appears throughout the AdS/CFT literature.

We stress that the vacuum subtraction protocol is \emph{only} valid for AdS spacetimes satisfying the relaxed Fefferman-Graham falloff \eqref{eq:relaxed-FG}. For spacetimes with slower matter falloffs, the areas of extremal surfaces in different spacetimes have different divergences, and cannot be compared; the vacuum-subtracted entropy of a subregion in any such spacetime is infinite. As discussed in Appendix \ref{app:FG-falloff}, these spacetimes are unusual in the sense that they have nonvanishing local energy in a neighborhood of the boundary. However, a class of holographic CFT states with state-dependent divergences in the entanglement entropy was constructed using field theory arguments in \cite{MW}; Theorem \ref{thm:state-dependence} implies that the spacetimes dual to these CFT states must violate equation \eqref{eq:relaxed-FG}.

That said, there are quantities that compare entanglement in different states that may be finite even when the vacuum-subtracted entropy is infinite. The \emph{quantum relative entropy}, for example, given by
\begin{equation}
    S(\rho || \sigma) \equiv \tr(\rho \log \rho) - \tr(\rho \log \sigma),
\end{equation}
can be written in terms of the modular Hamiltonian $H_{\sigma} = - \log \sigma$ as
\begin{equation} \label{eq:rel-ent}
    S(\rho || \sigma) = \langle H_{\sigma} \rangle_{\rho} - S(\rho).
\end{equation}
Both terms in this expression are divergent in the limit as the cutoff is removed. However, path integral arguments given in \cite{MW} suggest that the relative entropy of a subregion between two different holographic states is finite even when the states have different divergences in the entanglement entropy. One could check this analytically when $\sigma$ is a ball-shaped subregion of the CFT vacuum, since the modular Hamiltonian for such states has an explicit closed form \cite{BCHM}. In general, a more careful accounting of divergences in the modular Hamiltonian for arbitrary states would be required to verify the finiteness of \eqref{eq:rel-ent} from the general relativity side of the holographic dictionary.

\section{Discussion}
\label{sec:discussion}

In this paper, we have addressed issues involving radial cutoffs for the holographic entanglement entropy in AdS/CFT, and have shown that the holographic prescription for computing finite, information-theoretic combinations of entanglement entropies in a holographic CFT is cutoff-invariant. This result implies, in particular, that the notion of a ``globally minimal'' surface homologous to a given boundary region among several extremal candidates is independent of the cutoff used to compare their areas. We have also attempted to construct a unified toolkit for the covariant analysis of extremal surfaces, detailed in Appendix \ref{app:extremal-surfaces}. This kit consists of four equivalent, covariant ways of characterizing extremal surfaces, including the ``$k$-normal form divergence'' characterization used in this paper to prove the cutoff-covariance of holographic entanglement entropy. We hope these tools will prove broadly useful in the study of extremal surfaces both inside and outside of AdS/CFT.

Our main results are summarized as follows. First, any two extremal surfaces in an asymptotically anti-de Sitter spacetime that limit to the same boundary region have identical area divergences with respect to any choice of defining function $z$. Second, the finite quantity obtained by subtracting the formally divergent areas of two such surfaces is independent of the choice of defining function. Third, if one imposes the additional constraint of the relaxed Fefferman-Graham falloff condition \eqref{eq:relaxed-FG}, then the area divergence of an extremal surface with respect to a \emph{special defining function} (cf. Section \ref{sec:conformal-infinity}) is universal among all extremal surfaces anchored to the same boundary region in \emph{any} spacetime satisfying the falloff. Fourth, the finite difference in area between two such surfaces in \emph{different} spacetimes, a special case of which is the vacuum-subtracted entanglement entropy, is independent of the choice of special defining function $z$.

As a final comment, we note that our covariant techniques for analyzing extremal surfaces are all purely local. This means that they can be used to study any extremal surface, not just the ones that are ``globally minimal'' within a certain class. Such extremal surfaces appear in many interesting places in AdS/CFT: for example, as subleading saddles in the worldvolume action of a bulk brane \cite{M1998, GW1999}, or as probes of non-spatial entanglement in the study of holographic entwinement \cite{entwinement}. Another potentially interesting avenue for future work would be to apply the covariant methods of analysis presented in Appendix \ref{app:extremal-surfaces}, all of which are developed in terms of smooth deformations of hypersurfaces, to the ``shape deformations'' of entangling surfaces that have appeared in the study of the quantum focusing conjecture and the quantum null energy condition \cite{qnec1, qnec2, qnec3, qnec4, qnec5, qnec6, qnec7}.

\acknowledgments{I am greatly indebted to Bob Wald for his insistence on mathematically precise physics, without which this work would not exist. I am also very grateful to Matt Headrick, Veronika Hubeny, Adam Levine, Geoff Penington, Gautam Satishchandran, Antony Speranza, and Douglas Stanford for useful conversations, and especially to Sean Colin-Ellerin, Patrick Hayden, Marija Toma\u{s}evi\'{c}, and Bob Wald for comments on an early draft of this paper. I also thank Ben McKenna, an ever-useful mathematician contact, for pointing me to the proof of the matrix determinant lemma that appears in Appendix \ref{app:FG}. Major parts of this paper were completed at the Yukawa Institute for Theoretical Physics during workshop YITP-T-19-03 and at the Centro de Ciencias de Benasque during the 2019 workshop ``Gravity: new perspectives from strings and higher dimensions.'' This work was supported by AFOSR (FA9550-16-1-0082), the Simons It from Qubit collaboration, and DOE Award No. DE-SC0019380.}

\appendix

\section{Fefferman-Graham Revisited}
\label{app:FG}

In this appendix, we prove the generalized Fefferman-Graham theorem presented in Section \ref{sec:conformal-infinity}. For convenience, we restate this theorem here:

\setcounter{theorem}{0}
\begin{theorem}[Generalized Fefferman-Graham Expansion]
    \label{thm:FG-app}
    Let $z$ be a special defining function for a $d$-dimensional, asymptotically anti-de Sitter spacetime, with coordinates $x^j$ chosen along the level sets of $z$ so that the metric takes the form
    \begin{equation} \label{eq:FG-app}
        ds^2 = \frac{1}{z^2} \left[ dz^2 + \tilde{h}_{jk} d x^j dx^k \right].
    \end{equation}
    Suppose, further, that the spacetime metric $g_{ab}$ satisfies the metric falloff condition
    \begin{equation} \label{eq:app-thm-falloff}
        R_{ab} = - (d-1) g_{ab} + o(z^k)
    \end{equation}
    for some integer $k \geq -2.$ Then the smooth induced metric $\tilde{h}_{jk}$ is uniquely determined as a power series in $z$ up to and including order $z^{k+2}$ or $z^{d-2}$, whichever comes first. This power series is even in $z$.
\end{theorem}

\setcounter{theorem}{4}

Usually, this theorem is presented without allowing arbitrary falloff in equation \eqref{eq:app-thm-falloff}. The standard falloff condition imposed is the \emph{Fefferman-Graham falloff}, given by
\begin{equation} \label{eq:app-FG-falloff}
    R_{ab} = - (d-1) g_{ab} + o(z^{d-4}).
\end{equation}
However, the proof of the theorem is not any more difficult if one relaxes this condition to a generic $z^k$ falloff, and spacetimes with slower falloff have been studied in the AdS/CFT literature before \cite{MW}. As such, we prove the theorem in full generality in subsection \ref{app:FG-proof}. In subsection \ref{app:FG-falloff}, we comment on the typicality of the Fefferman-Graham falloff condition among physically reasonable spacetimes, and argue that it is fairly generic among spacetimes of physical interest.

\subsection{The Fefferman-Graham Expansion with Arbitrary Falloff}
\label{app:FG-proof}

Proving Theorem \ref{thm:FG-app} is as simple as writing out equation \eqref{eq:app-thm-falloff} in coordinates adapted to level sets of the defining function and matching both sides order-by-order in $z$. In any system of coordinates, the components of the Ricci tensor can be written in terms of Christoffel symbols as
\begin{equation}
    R_{\mu \nu}
        = \del_{\alpha} \Gamma^{\alpha}{}_{\mu \nu}
            - \del_{\mu} \del_{\nu} \ln\sqrt{|g|}
            + \Gamma^{\alpha}{}_{\mu \nu} \del_{\alpha} \ln\sqrt{|g|}
            - \Gamma^{\alpha}{}_{\mu \beta} \Gamma^{\beta}{}_{\nu \alpha},
\end{equation}
where $\del_{\mu}$ denotes a coordinate derivative and $g$ is the determinant of the metric. By solving for the Christoffel symbols using the metric form given in equation \eqref{eq:FG-app}, one can show that the components of the Ricci tensor satisfy the following equations:
\begin{eqnarray}
    R_{zz} & = & - \frac{d-1}{z^2} - \del_z \del_z \ln\sqrt{|\tilde{h}|} + \frac{1}{z}\, \del_z \ln\sqrt{|\tilde{h}|} + \frac{1}{4} (\del_z \tilde{h}^{\ell m})  (\del_z \tilde{h}_{\ell m}) \label{eq:ricci-zz} \\
    R_{zj} & = & \frac{1}{2}\, \del_{\ell} (\tilde{h}^{\ell m}\, \del_z \tilde{h}_{m j})
                - \del_j \del_z \ln\sqrt{|\tilde{h}|}
                + \frac{1}{2}\, \tilde{h}^{\ell m}\, (\del_z \tilde{h}_{j m})\, \del_{\ell} \ln\sqrt{|\tilde{h}|} \nonumber \\
                && - \frac{1}{2}\, (\tilde{h}^{\ell m}\, \del_z \tilde{h}_{m n})\, \tilde{\Gamma}^{n}{}_{j \ell} \label{eq:ricci-zj}\\
    R_{jk} & = & - \frac{d-1}{z^2}\, \tilde{h}_{jk} + \frac{d-2}{2 z}\, \del_z \tilde{h}_{jk}
                - \frac{1}{2}\, \del_z \del_z \tilde{h}_{jk}
                + \frac{1}{z}\, \tilde{h}_{jk}\, \del_z \ln\sqrt{|\tilde{h}|} \nonumber \\
                && - \frac{1}{2}\, (\del_z \tilde{h}_{jk})\, \del_z \ln\sqrt{|\tilde{h}|} + \frac{1}{2}\, (\del_z \tilde{h}_{\ell j})\, \tilde{h}^{\ell m}\, (\del_z \tilde{h}_{m k}) + \tilde{R}_{jk}. \label{eq:cutoff-ricci}
\end{eqnarray}
Here $z$ is the special defining function used to write the metric in the form \eqref{eq:FG-app}, and $j, k, \ell, m, n$ are indices that correspond to coordinates along the level sets of $z$. Several quantities in these expressions are associated to the smooth induced metric $\tilde{h}_{jk}$: $\tilde{R}_{jk}$ is its Ricci tensor, $\tilde{\Gamma}^{\ell}{}_{m n}$ are its Christoffel symbols, and $\tilde{h}$ is its metric determinant. We will only actually use the expression for the $jk$ components, equation \eqref{eq:cutoff-ricci}, since the $jk$ components of equation \eqref{eq:app-thm-falloff} suffice to specify $\tilde{h}_{jk}$ up to the desired order in $z$.

Matching the $jk$ components of \eqref{eq:app-thm-falloff} using the coordinate expression for the Ricci tensor given in \eqref{eq:cutoff-ricci} yields the following expression:
\begin{align} \label{eq:obyo-ricci}
    o(z^k)
            = & \frac{d-2}{2 z}\, \del_z \tilde{h}_{jk}
                - \frac{1}{2}\, \del_z \del_z \tilde{h}_{jk}
                + \frac{1}{z}\, \tilde{h}_{jk}\, \del_z \ln\sqrt{|\tilde{h}|} \nonumber \\
                & - \frac{1}{2}\, (\del_z \tilde{h}_{jk})\, \del_z \ln\sqrt{|\tilde{h}|} + \frac{1}{2}\, (\del_z \tilde{h}_{\ell j})\, \tilde{h}^{\ell m}\, (\del_z \tilde{h}_{m k}) + \tilde{R}_{jk}.
\end{align}
The smooth induced metric $\tilde{h}_{jk}$ can be expanded as a formal power series in $z$ as
\begin{equation}
    \tilde{h}_{jk}
        = \tilde{h}_{jk}^{(0)} + z \tilde{h}_{jk}^{(1)} + \dots + z^{n} \tilde{h}_{jk}^{(n)} + \dots,
\end{equation}
then plugged into equation \eqref{eq:obyo-ricci} order-by-order to solve for the induced metric. As a base case for an inductive procedure, we begin by checking the leading order, $1/z$ term of this expression. 

If the falloff order $k$ is $-2$, then Theorem \ref{thm:FG-app} claims that $\tilde{h}_{jk}$ is only uniquely determined at order $\tilde{h}_{jk}^{(0)}$. At this order, $\tilde{h}_{jk}$ is simply the conformal representative of the boundary metric picked out by our choice of defining function $z$. Terms of higher order cannot be determined using equation \eqref{eq:obyo-ricci}, as the leading order, $1/z$ term in the equation is undetermined. If, however, $k$ is greater than $-2$, then the $1/z$ term of \eqref{eq:obyo-ricci} reduces to
\begin{equation} \label{eq:pre-MDL}
    0
            = \frac{d-2}{2}\, \tilde{h}_{jk}^{(1)}
                + \frac{1}{2}\, \tilde{h}_{jk}^{(0)}\, \tilde{h}^{\ell m}_{(0)}\, \tilde{h}_{\ell m}^{(1)},
\end{equation}
where we have expanded terms involving the determinant of the smooth induced metric according to
\begin{equation} \label{eq:metric-det-der-expansion}
    \del_z \ln\sqrt{|\tilde{h}|}
        = \frac{1}{2} \tilde{h}^{\ell m} \del_z \tilde{h}_{\ell m}.
\end{equation}
The tensor equation given by \eqref{eq:pre-MDL} constitutes a linear system of equations for the components of $\tilde{h}_{\ell m}^{(1)}.$ To see this more clearly, we may rewrite equation \eqref{eq:pre-MDL} as
\begin{equation}
    0
            = \left[ \frac{d-2}{2}\, \delta_{j}{}^{\ell} \delta_{k}{}^{m}
                + \frac{1}{2}\, \tilde{h}_{jk}^{(0)}\, \tilde{h}^{\ell m}_{(0)}\,  \right] \tilde{h}_{\ell m}^{(1)}.
\end{equation}
The solution of $\tilde{h}_{\ell m}^{(1)}$ is uniquely determined whenever the matrix
\begin{equation}
    M_{jk}{}^{\ell m}
            = \frac{d-2}{2}\, \delta_{j}{}^{\ell} \delta_{k}{}^{m}
                + \frac{1}{2}\, \tilde{h}_{jk}^{(0)}\, \tilde{h}^{\ell m}_{(0)}
\end{equation}
is invertible, i.e., whenever its determinant is nonzero.

A priori, it may seem like a difficult problem to determine the invertiblity of this matrix. However, we note that $M_{ij}{}^{\ell m}$ is a matrix of the form
\begin{equation}
    \alpha I + \ket{v}\bra{w},
\end{equation}
where $I$ is the identity matrix, $\ket{v}$ is the vector of components of $\tilde{h}_{jk}^{(0)} / \sqrt{2}$, and $\bra{w}$ is the covector of components of $\tilde{h}^{\ell m}_{(1)}/\sqrt{2}.$ Matrices of this form have determinant given by the \emph{matrix determinant lemma}, a special case of Sylvester's determinant identity that states that any such matrix has determinant given by
\begin{equation}
    \det(\alpha I + \ket{v}\bra{w}) = \alpha^{n-1} (\alpha + \braket{w}{v}),
\end{equation}
where $n$ is the dimension of the vector space over which the matrix is defined. It follows that the determinant of $M_{jk}{}^{\ell m}$ is given by 
\begin{equation}
    \det(M_{jk}{}^{\ell m})
            = \left(\frac{d-2}{2}\right)^{\frac{d (d-1)}{2}} \left[ \frac{d-2}{2} + \frac{1}{2} \tilde{h}_{\ell m}^{(0)}\, \tilde{h}^{\ell m}_{(0)}\right],
\end{equation}
where the exponent $d (d-1)/2$ appears as the dimension of the space of symmetric tensors on a $(d-1)$-dimensional surface. Expanding the identity $\tilde{h}_{\ell m} \tilde{h}^{\ell m} = d-1$ perturbatively in $z$ yields the expression
\begin{equation}
    \tilde{h}_{\ell m}^{(0)}\, \tilde{h}^{\ell m}_{(0)} = d-1,
\end{equation}
and so the determinant of $M_{jk}{}^{\ell m}$ takes the form
\begin{equation}
    \det(M_{jk}{}^{\ell m})
            = \left(\frac{d-2}{2}\right)^{\frac{d (d-1)}{2}} \frac{2d-3}{2}.
\end{equation}
This determinant vanishes only for $d=2$ and $d=3/2,$ neither of which are satisfied by the spacetimes we consider. This implies that $M_{jk}{}^{\ell m}$ is invertible, and hence that $\tilde{h}_{jk}^{(1)}$ is uniquely determined to satisfy
\begin{equation}
    \tilde{h}_{jk}^{(1)} = 0
\end{equation}
whenever the AdS equation \eqref{eq:app-thm-falloff} is satisfied at order $1/z.$

With this base case established, we may now proceed to solve for $\tilde{h}_{jk}$ inductively at all orders. As an inductive assumption, we will assume that $\tilde{h}_{jk}^{(p)}$ is known for all $p < q$, and furthermore that $\tilde{h}_{jk}^{(p)}$ vanishes whenever $p < q$ is odd. This inductive assumption is satisfied by our base cases, in which $\tilde{h}_{jk}^{(0)}$ is specified by the choice of defining function and $\tilde{h}_{jk}^{(1)}$ was shown to vanish. At order $z^{q-2}$, equation \eqref{eq:obyo-ricci} takes the form
\begin{align} \label{eq:generic-Ricci-expansion}
    o(z^k)
            = z^{q-2} & \left\{ \frac{d-1-q}{2}\, q\, \tilde{h}_{jk}^{(q)}
                + \frac{1}{2} \sum_{\substack{A+B+C=q \\ C \geq 1}} C\, \tilde{h}_{jk}^{(A)}\, \tilde{h}^{\ell m}_{(B)}\,
                \tilde{h}_{\ell m}^{(C)} \right. \nonumber \\
                & \left. +\, \frac{1}{2}\,  \sum_{\substack{A+B+C=q \\ A,C \geq 1}} A\, C\, \left[ \tilde{h}_{\ell j}^{(A)}\, \tilde{h}^{\ell m}_{(B)}\, \tilde{h}_{m k}^{(C)}
                - \frac{1}{2} \tilde{h}_{jk}^{(A)}\, \tilde{h}^{\ell m}_{(B)}\, \tilde{h}_{\ell m}^{(C)} \right] + \tilde{R}_{jk}^{(q-2)} \right\},
\end{align}
where we have again expanded derivatives of the metric determinant according to \ref{eq:metric-det-der-expansion}. If $k$ is equal to $q-3$, then this equation cannot be used to determine $\tilde{h}_{jk}^{(q)}.$ As per the statement of Theorem \ref{thm:FG-app}, then, the induced metric is determined at most up to order $z^{k+2}$. Similarly, if $q$ is equal to $d-1$, then the first term on the right-hand side of \eqref{eq:generic-Ricci-expansion} vanishes and so $\tilde{h}_{jk}^{(q)}$ is undetermined by \eqref{eq:generic-Ricci-expansion} --- this is the source of the other claim in Theorem \ref{thm:FG-app}, that the metric is determined at most up to order $z^{d-2}$. If neither of these two failure conditions is met, however, we claim that \eqref{eq:generic-Ricci-expansion} can be used to solve for $\tilde{h}_{jk}^{(q)}.$

Collecting the terms proportional to $\tilde{h}_{jk}^{(q)}$ in \eqref{eq:generic-Ricci-expansion} yields the expression
\begin{align} \label{eq:obyo-solving}
    \frac{q}{2}\left[(d-1-q) \delta_{j}{}^{\ell} \delta_{k}{}^{m} + \tilde{h}_{jk}^{(0)} \tilde{h}^{\ell m}_{(0)} \right] \tilde{h}_{\ell m}^{(q)}
        = & -\frac{1}{2} \sum_{\substack{A+B+C=q \\ q-1 \geq C \geq 1}} C \tilde{h}_{jk}^{(A)} \tilde{h}^{\ell m}_{(B)}
                \tilde{h}_{\ell m}^{(C)} + \tilde{R}_{jk}^{(q-2)} \\
                & - \frac{1}{2}  \sum_{\substack{A+B+C=q \\ A,C \geq 1}} A C \left[ \tilde{h}_{\ell j}^{(A)} \tilde{h}^{\ell m}_{(B)} \tilde{h}_{m k}^{(C)}
                - \frac{1}{2} \tilde{h}_{jk}^{(A)} \tilde{h}^{\ell m}_{(B)} \tilde{h}_{\ell m}^{(C)} \right]. \nonumber
\end{align}
Terms of the form $\tilde{h}_{jk}^{(p)}$ appearing on the right-hand side of \eqref{eq:obyo-solving} are all of order $p < q$, and are assumed to be known by our inductive assumption. This expression constitutes a system of linear equations for $\tilde{h}_{jk}^{(q)},$ and by following exactly the same matrix determinant lemma argument given above, one can show that the matrix acting on $\tilde{h}_{\ell m}^{(q)}$ is invertible so long as $q < d-1$ is satisfied. It follows that $\tilde{h}_{jk}^{(q)}$ is determined by equation \eqref{eq:obyo-solving} under the inductive assumption that $\tilde{h}_{jk}^{(p)}$ is known for $p < q.$ Furthermore, when $q$ is odd, the terms in \eqref{eq:obyo-solving} involving sums over labels satisfying $A+B+C=q$ must vanish since at least one of $A$, $B$, or $C$ is odd, and we assumed that odd terms vanish for $p < q$. Similarly, one can check that the Ricci tensor of a metric with an even power series is itself an even power series, and so $\tilde{R}_{jk}^{(q-2)}$ vanishes for $q$ odd. The vanishing of the right-hand side of \eqref{eq:obyo-solving} implies that $\tilde{h}_{jk}^{(q)}$ vanishes when $q$ is odd, which validates our second inductive assumption.

We have thus proven that the $jk$ components of equation \eqref{eq:app-thm-falloff} uniquely determine $\tilde{h}_{jk}$ as an even power series in $z$ up to and including order $z^{d-2}$ or $z^{k+2},$ whichever comes first. In principle, one could go on to check that the solution we obtained by matching the $jk$ components of \eqref{eq:app-thm-falloff} is compatible with $zz$ and $zj$ components of the same equation; however, for our purposes, this is not actually necessary. We have assumed that we are given a spacetime metric that satisfies equation \eqref{eq:app-thm-falloff}, and asked how far into the bulk the induced metric on level sets of $z$ is determined by the boundary metric alone. This is a question of \emph{uniqueness}. Checking the compatibility of the $zz$ and $zj$ components with our answer would amount to the question of \emph{existence}: given an arbitrary metric on the boundary, does there exist a local extension into some ``bulk direction'' that is asymptotically AdS? This is an interesting question, but irrelevant in the context of Theorem \ref{thm:FG-app}. Discussions of the existence of such solutions can be found in the mathematics literature in \cite{FG1, FG2}.

\subsection{Comments on the Fefferman-Graham Falloff Condition}
\label{app:FG-falloff}

With Theorem \ref{thm:FG-app} proven in full generality, we now return to the Fefferman-Graham falloff condition \eqref{eq:app-FG-falloff}. As explained in Section \ref{sec:two-spacetimes}, spacetimes satisfying the Fefferman-Graham falloff have state-independent divergences in the holographic entanglement entropy; this makes them especially interesting in the context of holography. Here, we provide two comments on features of the Fefferman-Graham falloff in general relativity, both of which suggest that Fefferman-Graham falloff is fairly generic in ``physically reasonable'' spacetimes.

Our first observation is that the Fefferman-Graham falloff is exactly the falloff required for a family of local observers near the boundary to measure vanishing local energy. Using the general form of Einstein's equations with a (negative) cosmological constant,
\begin{equation}
    R_{ab} - \frac{1}{2} R g_{ab} + \Lambda g_{ab} = 8 \pi G_N T_{ab},
\end{equation}
one can show that the Ricci tensor always satisfies
\begin{equation}
    R_{ab} = - (d-1) g_{ab} - \frac{8 \pi G_N T}{d-2} g_{ab} + 8 \pi G_N T_{ab},
\end{equation}
where $T$ is the trace of the bulk stress-energy tensor. The Fefferman-Graham falloff condition is therefore satisfied if and only if the stress-energy tensor decays near the boundary as
\begin{equation} \label{eq:app-FG-energy-falloff}
    T_{ab} = o(z^{d-4}).
\end{equation}

The local energy measured by a family of timelike observers in a small spatial region $B$ is given by
\begin{equation} \label{eq:local-energy}
    \int_B T_{ab}\, u^a\, u^b\, \bm{\alpha}_{c_1 \dots c_{d-1}},
\end{equation}
where $u^a$ is the unit tangent vector describing the family of timelike observers and $\bm{\alpha}$ is the volume form on $B$. Since the metric scales as $1/z^2$, any normalized\footnote{Here we call a tensor ``normalized'' if its contraction with its dual is order one in $z$. For example, a tensor $P^{a}{}_{bc}$ is normalized if it satisfies $P^{a}{}_{bc} P_{a}{}^{bc} = O(1).$} tensor with $m$ up-indices and $n$-down indices scales as $z^{m-n}.$ In particular, the unit timelike vector field $u^a$ scales as $z$ and the volume form $\bm{\alpha}$ scales as $z^{-(d-1)}.$ These two scalings together imply that the local energy measured by a family of observers, \eqref{eq:local-energy}, diverges logarithmically in $z$ exactly when $T_{ab}$ scales as $z^{d-4}.$ For $T_{ab} = o(z^{d-4}),$ the Fefferman-Graham scaling, \eqref{eq:local-energy} vanishes near the boundary.\footnote{Here we are implicitly assuming that $T_{ab}$ is smooth in $z$ at $z=0$; otherwise, the  condition $T_{ab} = o(z^{d-4})$ would allow solutions like $T_{ab} \sim z^{d-4} / \ln{z}$, for which equation \eqref{eq:local-energy} still diverges near the boundary.}

We also note that it is generically rather difficult to engineer matter configurations that violate equation \eqref{eq:app-FG-energy-falloff}. Normalizable Klein-Gordon fields satisfying the scalar wave equation
\begin{equation}
    \Del_a \Del^a \phi - m^2 \phi^2 = 0,
\end{equation}
for example, scale as
\begin{equation} \label{eq:phi-scaling}
    \phi \sim z^{d-1}
\end{equation}
in any asymptotically anti-de Sitter spacetime whose boundary conditions allow $\phi$ to be expanded smoothly in $z$.\footnote{Equation \eqref{eq:phi-scaling} can be verified by writing the scalar wave equation in a Fefferman-Graham system of coordinates and solving it order-by-order in $z$ up to order $z^{d-1}$, much in the same way that we solved for the induced metric using the AdS equation in subsection \ref{app:FG-proof}.}\textsuperscript{,}\footnote{Classical scalar fields are not necessarily smooth in $z$ for arbitrary well-posed AdS boundary conditions. For a detailed analysis of this issue, see Section 3.2 of \cite{IW2004}.} The stress-energy tensor of a classical Klein-Gordon field scales as the derivative of the field squared, implying that the stress-energy scales as
\begin{equation}
    T_{ab} \sim z^{2(d-2)}.
\end{equation}
This is \emph{significantly} faster than the $z^{d-4}$ Fefferman-Graham falloff, exceeding it by a factor of $z^d.$ One could in principle perform a similar analysis using electromagnetic fields, Dirac fields, or any other classical field of interest. In general, though, it seems that rather unusual matter configurations are required to produce stress-energy tensors that violate the Fefferman-Graham falloff. For examples of matter configurations motivated by the AdS/CFT correspondence that \emph{do} violate the Fefferman-Graham falloff condition, see \cite{MW}.

\section{Covariant Analysis of Extremal Surfaces}
\label{app:extremal-surfaces}

Perhaps the most common characterization of an extremal surface in the physics literature is as a local saddle point of the area action. For a hypersurface $\Sigma$, one may write the induced metric $h_{j k}$ on $\Sigma$ in a particular set of internal coordinates, represent the area of $\Sigma$ formally as
\begin{equation}
    \mathrm{Area} = \int \sqrt{|h|}\,\, \dd x^{j},
\end{equation}
and call the surface \emph{extremal} if the Euler-Lagrange equations for $\sqrt{h}(x^j)$ are satisfied. This formulation of the extremality criterion has the advantage of being straightforward to compute for surfaces that can be easily parameterized, but the disadvantages of (i) not being explicitly covariant and (ii) losing the rich structure of the spacetime surrounding $\Sigma$ by writing everything in terms of the ``induced'' geometry on $\Sigma$ itself.

In this appendix, we present a more covariant, geometrically motivated approach to the study of extremal surfaces. In particular, we provide four equivalent, covariant ways to characterize extremality, each of which adds a different tool to our kit for analyzing extremal surfaces. We first provide a general definition of extremal surfaces in terms of smooth deformations of submanifolds, then provide two computationally useful characterizations of extremality in the language of differential forms, and finally relate our constructions to the vanishing trace of the extrinsic curvature. While some of these characterizations may be original --- namely, those involving differential forms, which we have not seen presented elsewhere in the form they take in this appendix --- it is likely that all results in this appendix may be found elsewhere in the literature. We present them here not as new results, but as a coherent ``manual'' for readers seeking to undertake their own covariant analyses of extremal surfaces.

\subsection{Defining Extremality}

Rather than defining a surface as ``extremal'' if it satisfies particular equations of motion, let us begin with the most \emph{intuitive} characterization of extremality: that a surface is extremal if its area does not increase at linear order under any compact deformation. In a Lagrangian framework, this is of course equivalent to the surface satisfying the equations of motion of the area action; however, we will see that beginning with a general framework of deformations will be useful in defining geometric tools for extremal surface analysis.

Consider a codimension-$k$ hypersurface $\Sigma_0$ in a $d$-dimensional spacetime. A \emph{one-parameter deformation} of $\Sigma_0$ is a smooth family of surfaces $\Sigma_{\lambda}$ that coincides with $\Sigma_0$ at $\lambda = 0$.\footnote{We will generally restrict our consideration to foliations with locally integrable normal bundles, and assume that all surfaces in question admit such foliations. This assumption makes our proofs simpler, but is not actually necessary for proving the main results of the paper; we will comment further on these subtleties as they arise. For a detailed analysis of the properties of submanifolds that do \emph{not} admit such foliations, see Section 3 of \cite{speranza2019}.} A \emph{compact one-parameter deformation} is a deformation that only affects some compact subregion of $\Sigma_0$; formally, we require that all surfaces $\Sigma_{\lambda}$ intersect on a subregion of $\Sigma_0$ whose complement is compact. Since the deformation is compact, the difference in area between $\Sigma_0$ and $\Sigma_{\lambda}$ is finite --- we may assign this finite difference to a function
\begin{equation}
    \Delta \mathrm{Area} (\lambda) \equiv \mathrm{Area}(\Sigma_{\lambda}) - \mathrm{Area}(\Sigma_0)
\end{equation}
As a basic definition of extremality, we propose the following:
\begin{definition}[One-Parameter Extremality]
    A smooth hypersurface $\Sigma_0$ is \emph{extremal} if any compact one-parameter deformation of
    $\Sigma_0$ satisfies
    \begin{equation}
        \Delta \mathrm{Area} (\lambda) = O(\lambda^2),
    \end{equation}
    i.e. if the linear contribution to the change in area vanishes.
    \label{def:1paramext}
\end{definition}

More generally, a codimension-$k$ hypersurface admits a $k$-parameter family of deformations $\Sigma_{\lambda_1, \dots, \lambda_k}$. We may equivalently adapt the above definition with these deformations in mind:
\begin{definition}[$k$-Parameter Extremality]
    A smooth, codimension-$k$ hypersurface $\Sigma_0$ is \emph{extremal} if any compact $k$-parameter deformation of $\Sigma_0$ satisfies
    \begin{equation}
        \Delta \mathrm{Area} (\lambda_1, \dots, \lambda_k) = O(\lambda_i \lambda_j)|_{i, j \in \{1, \dots, k\}}.
    \end{equation}
    \label{def:kparamext}
\end{definition}

Definition \ref{def:1paramext} trivially implies Definition \ref{def:kparamext}, since any $k$-parameter family of deformations can be reduced to $k$ one-parameter families by fixing all but one of the parameters $\lambda_i$ to zero. Definition \ref{def:kparamext} implies Definition \ref{def:1paramext} since any one-parameter deformation of a codimension-$k$ surface can be extended to a $k$-parameter deformation by picking $k-1$ directions normal to $\Sigma_0$ that form a linearly independent set with the original deformation vector $(\partial / \partial \lambda)^a$ and deforming along their geodesics.

This ``deformation'' characterization lends itself naturally to a characterization of extremal surfaces in the language of differential forms, which is the subject of the following subsection.

\subsection{Induced Volumes and $k$-Normal Forms}

Let us consider exactly \emph{how} the area of a codimension-$k$ surface changes along a $k$-parameter deformation. A codimension-$k$ surface in a $d$-dimensional spacetime comes equipped with an \emph{induced volume form} $\bm{\alpha}_{a_1 \dots a_{d-k}}$. The tangent bundle of such a surface is $(d-k)$-dimensional, and so it admits a unique $(d-k)$-form up to orientation and normalization. The induced volume form is specified by choosing an (arbitrary) orientation and requiring the normalization
\begin{equation} \label{eq:induced-vol-form}
    \bm{\alpha}_{a_1 \dots a_{d-k}} \bm{\alpha}^{a_1 \dots a_{d-k}} = (-1)^s\, (d-k)!,
\end{equation}
where $s$ is the number of minus signs in the signature of the induced metric on the tangent bundle.

A $k$-parameter deformation of a codimension-$k$ hypersurface $\Sigma_0$ defines a field of induced volume forms in a neighborhood of the surface (or rather, in a neighborhood of the subregion of $\Sigma_0$ that is not fixed by the deformation). The change in area between $\Sigma_0$ and another member of the deformation is given by
\begin{equation}
    \Delta \mathrm{Area} (\lambda_1, \dots, \lambda_k)
        = \int_{\Sigma_{\lambda_1, \dots, \lambda_k}} \bm{\alpha} - \int_{\Sigma_0} \bm{\alpha}.
\end{equation}
By Stokes' theorem, this difference can be computed by integrating the exterior derivative $\dd \bm{\alpha}$ over a spacetime region whose boundary is $\bar{\Sigma}_0 \cup \bar{\Sigma}_{\lambda_1, \dots \lambda_k}$, where here the bars over the surfaces indicate that we are considering only the compact subregions where the surfaces differ. If we label this region $S$, then the change in area is given by
\begin{equation} \label{eq:k-param-area-var}
    \Delta \mathrm{Area} (\lambda_1, \dots, \lambda_k)
        = \int_{S} \dd \bm{\alpha}.
\end{equation}
Since $\Sigma_{\lambda_1, \dots \lambda_k}$ and $\Sigma_0$ differ on only a compact set, $S$ itself can be chosen to be compact. Its volume is given by some function of the parameters $\lambda_i$ that vanishes in the limit $\lambda_i \rightarrow 0$. It follows that any $O(\lambda_i)$ changes in the form $\dd \bm{\alpha}$ contribute to the integral in \eqref{eq:k-param-area-var} above linear order; the linear contributions to \eqref{eq:k-param-area-var}, therefore, are proportional to terms of the form
\begin{equation}
    \int_{\bar{\Sigma}_0} n \cdot \dd \bm{\alpha},
\end{equation}
where $n$ is a vector that is normal to $\Sigma_0.$ But if the surface is extremal, then this must vanish! The integral of $\dd \bm{\alpha}$ over $\bar{\Sigma}_0$ vanishes if and only if (i) $\dd \bm{\alpha}$ vanishes, or if (ii) $\dd \bm{\alpha}$ has two components normal to $\Sigma_0.$ The second condition is forbidden by our assumption that directions normal to $\Sigma_0$ are locally integrable.\footnote{This claim follows from Lemma \ref{lem:div-knormal-normal}, which we will prove shortly. When the normal bundle is integrable, the Hodge dual of $\dd \bm{\alpha}$ is entirely normal to $\Sigma_0$ --- this implies that $\dd \bm{\alpha}$ has at most one component normal to $\Sigma_0$.} We conclude that a necessary and sufficient condition for extremality is that the induced volume form $\bm{\alpha}$ satisfies
\begin{equation} \label{eq:ext-der-extrem}
    \dd \bm{\alpha}|_{\Sigma_0} = 0
\end{equation}
for any $k$-parameter deformation of $\Sigma_0.$ This gives our second characterization of extremality, and our first in terms of differential forms.\footnote{In the mathematics literature, a closed differential form on an extremal surface is sometimes called a \emph{calibration}. Calibrations were applied to the study of holographic entanglement entropy in static spacetimes in \cite{BDGCY}.}

The induced volume form on $\Sigma_0$, however, can be unwieldy to use in computations. Its exterior derivative, in particular, is \emph{a priori} dependent on the choice of deformation, and can be annoying to compute. We shall find it significantly more tractable from a computational perspective to reframe equation \eqref{eq:ext-der-extrem} in terms of the \emph{unit $k$-normal form}.

For any codimension-$k$ surface $\Sigma_0,$ we may define the \emph{normal bundle} $\mathrm{N}(\Sigma_0)$ as the collection of all tangent directions in the spacetime manifold that are normal to $\Sigma_0.$ For simplicity, we assume in the following that $\Sigma_0$ is not a null surface, i.e., that no vectors are simultaneously tangent and normal to $\Sigma_0$. The normal bundle of a codimension-$k$ surface is $k$-dimensional, and hence admits a unique $k$-form up to orientation and normalization. Fixing an arbitrary orientation, we define the \emph{unit $k$-normal form} as the unique $k$-form $N_{a_1 \dots a_k}$ satisfying
\begin{equation} \label{eq:k-normal-form}
    N_{a_1 \dots a_k} N^{a_1 \dots a_k} = (-1)^s\, k!,
\end{equation}
where $s$ is the number of minus signs in the signature of the normal bundle.\footnote{Note that for a codimension-$1$ surface, $\mathbf{N}$ is just the unit normal vector.} The normalization of $\mathbf{N}$ is chosen so that the induced volume form on $\Sigma_0$ is given by the \emph{Hodge dual} of $\mathbf{N}$, i.e.
\begin{equation}
    \bm{\alpha}_{a_1 \dots a_{d-k}}
        = \frac{1}{k!}\, \epsilon_{a_1 \dots a_{d-k} a_{d-k+1} \dots a_{d}}
            N^{a_{d-k+1} \dots a_{d}},
\end{equation}
where $\bm{\epsilon}$ is the volume form on the full $d$-dimensional spacetime. Equivalently, comparing equation \eqref{eq:k-normal-form} to equation \eqref{eq:induced-vol-form}, we can see that the unit $k$-normal form is just the induced volume form on integral submanifolds of the normal bundle $\mathrm{N}(\Sigma_0).$ Since $\mathbf{N}$ is a $k$-form on the normal bundle, it must be proportional to any other $k$-form on the normal bundle, with the constant of proportionality fixed by normalization and orientation. In particular, it can always be written in the form
\begin{equation} \label{eq:kform-basis}
    N_{a_1 \dots a_k}
        = \pm\, k!\, n^{(1)}{}_{[a_1} n^{(2)}{}_{a_2} \dots n^{(k)}{}_{a_k]}
\end{equation}
for any basis $\{n_{(i)}{}^a\}$ of the normal bundle satisfying the following conditions:
\begin{enumerate}[(i)]
    \item If $n_{(i)}{}^a$ is spacelike or timelike, then it satisfies $n_{(i)}{}^{a} n^{(j)}{}_{a} = \pm \delta_{ij}.$ (I.e., it is orthogonal to every other vector in the basis, and is normalized to $\pm 1$.)
    \item If two \emph{distinct} basis vectors $n_{(i)}{}^a$ and $n_{(j)}{}^a$ are both null, then they satisfy $n_{(i)}{}^{a} n^{(j)}{}_{a} = \pm 1.$
\end{enumerate}
The sign in equation \eqref{eq:kform-basis} depends on the chosen orientation of $\Sigma_0$ relative to the orientation of $\{n_{(i)}^{a}\}$; it can be reversed by switching the order of any two vectors in the basis.

Having identified the induced volume form $\bm{\alpha}$ as the Hodge dual ${\ast}\,\mathbf{N}$, we may use the \emph{divergence identity}
\begin{equation}
    {\ast}\,\dd \bm{\alpha} = {\ast}\,\dd\,{\ast}\,\mathbf{N} = \Del^b N_{b a_2 \dots a_k}.
\end{equation}
Vanishing of $\dd \bm{\alpha}$, as per \eqref{eq:ext-der-extrem}, is therefore equivalent to the vanishing of the divergence of the $k$-normal form. This gives us a more computationally useful characterization of extremality: a codimension-$k$ surface $\Sigma_0$ is extremal if and only if its unit $k$-normal form satisfies\footnote{If one relaxes the assumption that the normal bundle of $\Sigma_0$ is locally integrable, and thus allows $\dd \bm{\alpha}$ to have nonvanishing terms with two components normal to $\Sigma_0$, then one can only guarantee that the normal components of \eqref{eq:div-norm-extrem} vanish.}
\begin{equation} \label{eq:div-norm-extrem}
    \nabla_b N^{b a_2 \dots a_k}|_{\Sigma_0} = 0
\end{equation}
for any deformation of $\Sigma_0$. This expression has a major advantage over \eqref{eq:ext-der-extrem} in that the divergence of a $k$-form is much more straightforward to compute than the exterior derivative of a $(d-k)$-form. Equation \eqref{eq:div-norm-extrem} is our third characterization of extremality, and the one with the greatest computational power. In fact, it is the condition used in the main body of the paper, in Section \ref{sec:main-sec}, to prove cutoff-covariance of the holographic entanglement entropy.

Still, we are left with a bit of a puzzle from these two characterizations of extremality in terms of differential forms. Namely, both seem at least naively to depend on a choice of deformation --- the exterior derivative of $\bm{\alpha}$, and likewise the divergence of $\mathbf{N}$, cannot be computed without defining these forms in a neighborhood of $\Sigma_0$. Checking extremality by checking equations \eqref{eq:ext-der-extrem} or \eqref{eq:div-norm-extrem} would seem to require checking them for \emph{all possible deformations} of $\Sigma_0.$ Not only does this seem prohibitively difficult in practical applications, it also seems at odds with our intuition that any characterization of extremality should be a local property of a surface's embedding in spacetime, not a property of its deformations.

We address this puzzle in the following subsection, in which we show that the quantity $\Del_a N^{a \dots}|_{\Sigma_0}$ can be written in terms of the trace of the extrinsic curvature of $\Sigma_0$, which manifestly depends only on the local geometry of $\Sigma_0$ itself. This implies that $\Del_a N^{a \dots}|_{\Sigma_0}$ is an invariant of the surface, and is thus independent of the choice of deformation. This connection between our ``differential forms picture'' of extremality and the more familiar condition of ``vanishing trace of the extrinsic curvature'' constitutes our fourth and final characterization of extremality.

\subsection{Extrinsic Curvature}

Formally, the extrinsic curvature tensor of $\Sigma_0$ at a point $p$ is a map from two copies of the tangent space at $p$ to a single copy of the normal space at $p$. It is a tensor $K^c{}_{ab}$ with two down-indices that belong to the cotangent bundle of the surface $\Sigma_0$ and one up-index that belongs to the normal bundle, defined so that for any two tangent vectors $X^a, Y^a,$ we have
\begin{equation} \label{eq:extrinsic-curvature}
    K^c{}_{ab} X^a Y^b = P^{c}{}_{b} X^{a} \Del_a Y^b,
\end{equation}
where $P^{a}{}_{b}$ is the orthogonal projector from the full spacetime tangent bundle down onto the hypersurface normal bundle $\mathrm{N(\Sigma_0)}.$\footnote{In the mathematics literature, the extrinsic curvature is frequently written as $(\Del_X Y)^{\perp}$ to emphasize that it gives the normal component of the directional derivative of $Y^a$ with respect to $X^a.$} (Equivalently, $P_{ab}$ is the induced metric on $\mathrm{N}(\Sigma_0)$.) Even though a covariant derivative appears in the right-hand side of the above expression, it can be evaluated without needing to specify $Y^a$ off of the surface, since the full expression takes the form of a directional derivative along another tangent direction $X^a$. We may exploit the product rule to rewrite \eqref{eq:extrinsic-curvature} in the manifestly symmetric form
\begin{equation} \label{eq:extrinsic-curvature-sym}
    K^c{}_{ab} X^a Y^b = P^{c}{}_{b} \Del_a (X^{a} Y^b),
\end{equation}
where the other product-rule term, $P^{c}{}_{b} Y^b \Del_a X^a$, vanishes since $Y^a$ tangent to $\Sigma_0$ implies $P^{c}{}_{b} Y^b = 0.$ By taking the trace of the extrinsic curvature over the two indices that belong to the cotangent bundle of $\Sigma_0$, $K^c{}_{ab}$ gives rise to a preferred vector on $\mathrm{N}(\Sigma_0)$ as
\begin{equation}
    K^{c}{}_{ab} \rightarrow K^{c}{}_{ab} h^{ab},
\end{equation}
where $h_{ab}$ is the induced metric on $\Sigma_0.$

This vector, $K^{c}{}_{ab} h^{ab}$, is intimately related to the divergence of the unit $k$-normal form: in fact, for any deformation of $\Sigma_0$, the divergence of \textbf{N} is proportional to \emph{the Hodge dual of $K^{c}{}_{ab} h^{ab}$ with respect to the induced metric on the normal bundle.} In other words, knowing the trace of the extrinsic curvature $K^{c}{}_{ab} h^{ab}$ on $\Sigma_0$ suffices to specify the quantity $\Del_a N^{a\dots}|_{\Sigma_0}$ for an arbitrary deformation, and vice versa. To show this, we first must verify the following lemma: that the divergence of $\mathbf{N}$ is a $(k-1)$-form on the normal bundle of $\Sigma_0.$

\begin{lemma} \label{lem:div-knormal-normal}
    For any $k$-parameter deformation $\Sigma_{\lambda_1, \dots \lambda_k}$ of a codimension-$k$ surface $\Sigma_0$, the divergence of the unit $k$-normal form,
    \begin{equation}
        \Del_a\, N^{a \dots},
    \end{equation}
    has support only along directions normal to the level sets of the deformation. In particular, its restriction to the undeformed surface,
    \begin{equation}
        \Del_a\, N^{a \dots}|_{\Sigma_0},
    \end{equation}
    is supported in the normal bundle $\mathrm{N}(\Sigma_0).$
\end{lemma}

\begin{proof}
Let $T^a$ be an arbitrary vector field tangent to the level sets of the deformation and contract $T^a$ into the first free index of the divergence. Performing this contraction and writing $\mathbf{N}$ in terms of a local basis as in \eqref{eq:kform-basis} yields the expression
\begin{equation} \label{eq:tangent-comp}
    T_{b}\, \Del_{a}\, N^{ab \dots}
        = k!\, T_{b}\, \Del_{a}\, n_{(1)}{}^{[a} n_{(2)}{}^{b} \dots n_{(k)}{}^{a_k]}.
\end{equation}
The derivative in the right-hand side of this expression can be expanded according to the product rule. Since each of the basis vectors $n_{(i)}{}^a$ is orthogonal to $T^a$, this expansion vanishes whenever the up-index $b$ ends up outside of the derivative. By carrying out the antisymmetrization in the up-indices, the remaining terms can be grouped so that each one is proportional to an expression of the form
\begin{equation} \label{eq:frob-commutator}
    T_b(n_{(i)}{}^{a} \Del_a n_{(j)}{}^{b} - n_{(j)}{}^{a} \Del_a n_{(i)}{}^{b})
        \equiv T_b (\mathscr{L}_{n_{(i)}} n_{(j)})^b,
\end{equation}
where $\mathscr{L}_{X} Y$ denotes the Lie derivative of $Y$ with respect to $X$, or, equivalently, the vector commutator $[X, Y]$. These terms must all vanish as a consequence of \emph{Frobenius' theorem}, which states that a smooth specification of subspaces $\mathcal{V}$ is integrable if and only if $\mathscr{L}_X Y \in \mathcal{V}$ holds for any $X, Y \in \mathcal{V}$. The normal subspaces of the deformation are all integrable by assumption, and so the commutator of any two normal vectors is itself a normal vector. It follows that terms of the form \eqref{eq:frob-commutator} vanish, and hence that the entire expression \eqref{eq:tangent-comp} vanishes. Since the divergence of $\mathbf{N}$ is antisymmetric in its $(k-1)$ free indices, it vanishes whenever a vector tangent to the level sets of the deformation is contracted into \emph{any} of its free indices. It follows that
\begin{equation}
    \Del_{a}\, N^{a \dots}
\end{equation}
has no support along the tangent directions of the level sets of the deformation, and hence that its restriction to $\Sigma_0$ is supported in the normal bundle $\mathrm{N}(\Sigma_0),$ as desired.
\end{proof}

\vspace{0.4cm}

We have now shown that the divergence of $\mathbf{N}$ is a $(k-1)$-form on the normal bundle $\mathrm{N}(\Sigma_0).$ The trace of the extrinsic curvature, by contrast, is a vector field on the normal bundle --- or, equivalently, a one-form. On a $k$-dimensional subspace such as the normal bundle, there is a natural duality between one-forms and $(k-1)$-forms provided by the Hodge dual. A one-form $\bm{\omega}$ can be mapped to a $(k-1)$ form as\footnote{The Hodge dual is usually defined on a $d$-dimensional geometry by mapping an $m$-form to a $(d-m)$-form through contraction with the $d$-index volume form on that geometry. Since $\mathbf{N}$ \emph{is} the $k$-index volume form on the normal bundle (cf. equation \eqref{eq:k-normal-form}), this is the same definition as given in \eqref{eq:hodge-dual-1} and \eqref{eq:hodge-dual-2}.}
\begin{equation} \label{eq:hodge-dual-1}
    (\ast\, \bm{\omega})_{a_1 \dots a_{k-1}} = N_{a_1 \dots a_k} \bm{\omega}^{a_k},
\end{equation}
and a $(k-1)$-form $\bm{\tau}$ can be mapped to a one-form as
\begin{equation} \label{eq:hodge-dual-2}
    (\ast\, \bm{\tau})_{a_1} = \frac{1}{(k-1)!} N_{a_1 \dots a_k} \bm{\tau}^{a_2 \dots a_k}.
\end{equation}
We claim that this is exactly the relationship between the divergence of $\mathbf{N}$ and the trace of the extrinsic curvature --- that up to multiplication by a constant, the trace of the extrinsic curvature is the Hodge dual of the divergence of $\mathbf{N}.$ This is made precise in the following theorem.

\begin{theorem} \label{thm:kform-hodge}
    Let $\Sigma_{\lambda_1, \dots, \lambda_k}$ be a $k$-parameter deformation of a codimension-$k$ surface $\Sigma_0$. Then the divergence of the unit $k$-normal form $\mathbf{N}$ and the extrinsic curvature of $\Sigma_0$ are related on $\Sigma_0$ by
    \begin{equation} \label{eq:ext-curv-trace}
        K_{cab} h^{ab} 
            = - (-1)^s\, (k-1)!\, N_{c a_2 \dots a_k} \Del_{b} N^{b a_2 \dots a_k},
    \end{equation}
    where $s$ is the number of minus signs in the signature of the normal bundle and $h_{ab}$ is the induced metric on $\Sigma_0.$
\end{theorem}

\begin{proof}
    Recall from equation \eqref{eq:extrinsic-curvature-sym} that the extrinsic curvature of $\Sigma_0$ is defined so that
    \begin{equation}
        K_{cab} X^a Y^b = P_{cb} \Del_a X^a Y^b
    \end{equation}
    holds for any tangent vector fields $X^a$ and $Y^a$. (Here, we have lowered the index $c$ to make contact with the equation in the statement of the theorem.) Since the inverse of the induced metric, $h^{ab}$, can be written as a linear combination of terms of the form $X^a Y^b,$ it follows that the trace of the extrinsic curvature takes the form
    \begin{equation} \label{eq:trace-interm}
        K_{cab} h^{ab} = P_{cb} \Del_a h^{ab}.
    \end{equation}
    The induced metric on $\Sigma_0,$ $h_{ab}$, can be written in terms of the full spacetime metric and the projector onto the normal bundle as
    \begin{equation}
        h_{ab} = g_{ab} - P_{ab}.
    \end{equation}
    Plugging this into \eqref{eq:trace-interm} and using the fact that the covariant derivative of the spacetime metric vanishes yields the expression
    \begin{equation} \label{eq:trace-pen}
        K_{cab} h^{ab} = - P_{cb} \Del_a P^{ab}.
    \end{equation}
    
    We may now exploit the fact that the projector onto the normal bundle can be written in terms of the unit $k$-normal form as
    \begin{equation} \label{eq:projector-knormal}
        P^{a}{}_{b} = (-1)^s\, (k-1)!\, N^{a d_2 \dots d_k} N_{b d_2 \dots d_k},
    \end{equation}
    where $s$ is the number of minus signs in the signature of the normal bundle.\footnote{Equation \eqref{eq:projector-knormal} is a specific consequence of the more general principle that any contraction of a volume form with itself can be written in terms of the associated metric; see Appendix B of \cite{WaldBook} for a review.} Making this substitution into the \emph{second} projector that appears in \eqref{eq:trace-pen} and expanding according to the product rule yields the expression
    \begin{equation} \label{eq:trace-fin}
        K_{cab} h^{ab} = - (-1)^s\, (k-1)!\,
                            \left[ N_{c d_2 \dots d_k} \Del_a N^{a d_2 \dots d_k}
                                    + N^{a d_2 \dots d_k} P_{c}{}^{b} \Del_a N_{b d_2 \dots d_k} \right].
    \end{equation}
    The first term in \eqref{eq:trace-fin} is exactly what was claimed in the statement of the theorem. All that remains is to show that the second term vanishes.
    
    To see that this term vanishes, we write $\mathbf{N}$ in terms of an orthonormal basis of timelike and spacelike vectors, $\{n_{(i)}{}^{a}\}$, as in equation \eqref{eq:kform-basis}. For any such basis, the projector onto the normal bundle takes the form
    \begin{equation}
        P_{c}{}^{b} = \sum_j \sigma_j\, n^{(j)}{}_{c} n_{(j)}{}^{b},
    \end{equation}
    where $\sigma_j = n_{(j)}{}^{a} n^{(j)}{}_{a}$ is the causal sign of the vector in question. The factor $P_{c}{}^{b} \Del_a N_{b d_2 \dots d_k}$, which appears in the second term of \eqref{eq:trace-fin}, can then be written in this basis as
    \begin{equation}
        P_{c}{}^{b} \Del_a N_{b d_2 \dots d_k}
            = k!\, \sum_j \sigma_j\, n^{(j)}{}_{c} n_{(j)}{}^{b}
                \Del_a\, n^{(1)}{}_{[b} n^{(2)}{}_{d_2} \dots n^{(k)}{}_{d_k]}.
    \end{equation}
    One can reorder the vectors inside the derivative to bring the term $n^{(j)}{}_{a}$ to the front by incurring a sign change $(-1)^{j+1}.$ This yields the expression
    \begin{equation}
        P_{c}{}^{b} \Del_a N_{b d_2 \dots d_k}
            = k!\, \sum_j (-1)^{j+1}\, \sigma_j\, n^{(j)}{}_{c} n_{(j)}{}^{b}
                \Del_a\, n^{(j)}{}_{[b} n^{(1)}{}_{d_2} \dots n^{(k)}{}_{d_k]}.
    \end{equation}
    We can then use the product rule to pull $n^{(j)}{}_{a}$ out of the derivative and write this expression in the form
    \begin{equation}
        P_{c}{}^{b} \Del_a N_{b d_2 \dots d_k}
            = k!\, \sum_j (-1)^{j+1}\, \sigma_j\, n^{(j)}{}_{c} n_{(j)}{}^{b}
                n^{(j)}{}_{[b} \Del_a\, n^{(1)}{}_{d_2} \dots n^{(k)}{}_{d_k]},
    \end{equation}
    where the other terms in the product rule vanish due to orthonomality of the vectors $\{n_{(i)}^{a}\}$. With some careful accounting, one can check that this term may be rewritten as
    \begin{align} \label{eq:projector-pen}
        P_{c}{}^{b} \Del_a N_{b d_2 \dots d_k}
            =\, & (k-1)!\, \sum_j (-1)^{j+1}\, \sigma_j\, n^{(j)}{}_{c}\, \times\\
                & (\delta^{e_2}{}_{[d_2} - \sigma_j n_{(j)}{}^{e_2}
                n^{(j)}{}_{[d_2}) \dots (\delta^{e_k}{}_{d_k]} - \sigma_j n_{(j)}{}^{e_k}
                n^{(j)}{}_{d_k]})
                \Del_a\, n^{(1)}{}_{e_2} \dots n^{(k)}{}_{e_k}. \nonumber 
    \end{align}
    
    Consider now the factors in this expression of the form
    \begin{equation} \label{eq:subtract-off}
        (\delta^{e_2}{}_{d_2} - \sigma_j n_{(j)}{}^{e_2}
                n^{(j)}{}_{d_2}).
    \end{equation}
    The term
    \begin{equation}
        \sigma_j n_{(j)}{}^{e_2} n^{(j)}{}_{d_2}
    \end{equation}
    is simply the projector onto the one-dimensional subspace spanned by $n_{(j)}{}^a$. By subtracting it off in \eqref{eq:subtract-off}, we are essentially constructing the projector onto \emph{the orthogonal complement of $n_{(j)}{}^{a}$.} For convenience, we label this projector
    \begin{equation}
        Q_{(j)}^{e_2}{}_{d_2} \equiv (\delta^{e_2}{}_{d_2} - \sigma_j n_{(j)}{}^{e_2}
                n^{(j)}{}_{d_2}).
    \end{equation}
    With this notation, equation \eqref{eq:projector-pen} can be simplified as
    \begin{equation} \label{eq:projector-fin}
        P_{c}{}^{b} \Del_a N_{b d_2 \dots d_k}
            =\, (k-1)!\, \sum_j (-1)^{j+1}\, \sigma_j\, n^{(j)}{}_{c}
                Q_{(j)}^{e_2}{}_{[d_2} \dots Q_{(j)}^{e_k}{}_{d_k]}
                \Del_a\, n^{(1)}{}_{e_2} \dots n^{(k)}{}_{e_k}.
    \end{equation}
    Now, let us return to the extraneous term in equation \eqref{eq:trace-fin} that we wish to show vanishes. Using \eqref{eq:projector-fin}, we may now write it in the form
    \begin{align} \label{eq:this-appendix-is-finally-done}
        N^{a d_2 \dots d_k} P_{c}{}^{b} \Del_a N_{b d_2 \dots d_k}
            = & (k-1)!\, \sum_j (-1)^{j+1}\, \sigma_j\, n^{(j)}{}_{c} \times \\
                & N^{a d_2 \dots d_k} Q_{(j)}^{e_2}{}_{d_2} \dots Q_{(j)}^{e_k}{}_{d_k}
                \Del_a\, n^{(1)}{}_{e_2} \dots n^{(k)}{}_{e_k}. \nonumber
    \end{align}
    $\mathbf{N}$ is supported in the normal bundle of $\Sigma_0,$ and each of the projectors $Q_{(j)}$ projects the corresponding index onto the orthogonal complement of $n_{(j)}^{a}$. It follows from the antisymmetry of $\mathbf{N}$ that the only nonvanishing term in $N^{a d_2 \dots d_k} Q_{(j)}^{e_2}{}_{d_2} \dots Q_{(j)}^{e_k}{}_{d_k}$ must be proportional to $n_{(j)}^{a}$, i.e., it must be of the form
    \begin{equation}
        N^{a d_2 \dots d_k} Q_{(j)}^{e_2}{}_{d_2} \dots Q_{(j)}^{e_k}{}_{d_k}
            \propto n_{(j)}^{a} n_{(1)}{}^{[e_2} \dots n_{(k)}{}^{e_k]}.
    \end{equation}
    Plugging this back into \eqref{eq:this-appendix-is-finally-done}, it follows immediately from orthonormality of the basis that the entire expression vanishes. Returning to equation \eqref{eq:trace-fin}, we see therefore that the trace of the extrinsic curvature takes the form
    \begin{equation}
        K_{cab} h^{ab} = - (-1)^s\, (k-1)!\,
                            N_{c d_2 \dots d_k} \Del_a N^{a d_2 \dots d_k},
    \end{equation}
    as desired.
\end{proof}

\vspace{0.4cm}

Theorem \ref{thm:kform-hodge} implies that the trace of the extrinsic curvature vanishes if and only if the divergence of the $k$-normal form vanishes on $\Sigma_0.$ This equivalence constitutes our fourth and final characterization of extremality: a surface is extremal if and only if the trace of its extrinsic curvature vanishes. Theorem \ref{thm:kform-hodge} also implies an immediate, useful corollary: that the divergence of the unit $k$-normal form, restricted to $\Sigma_0$, is independent of the deformation of $\Sigma_0.$\footnote{A final advantage of Theorem \ref{thm:kform-hodge} is that it allows us to relax the assumption that the normal bundle of an extremal surface is locally integrable. Even when the normal bundle is not locally integrable, and the divergence of $\mathbf{N}$ does not vanish, the normal components of the divergence vanish and so equation \eqref{eq:ext-curv-trace} vanishes. All of the main results in this paper could be obtained equally well in this more general case, albeit with slightly more unwieldy expresions, by using the vanishing of \eqref{eq:ext-curv-trace} instead of the vanishing divergence of $\mathbf{N}.$}

\begin{corollary} \label{cor:corollary}
    Let $\Sigma_{\lambda_1, \dots, \lambda_k}$ be a $k$-parameter deformation of a codimension-$k$ surface $\Sigma_0$. Then the divergence of the unit $k$-normal form on $\Sigma_0$,
    \begin{equation}
        \Del_a N^{a \dots}|_{\Sigma_0},
    \end{equation}
    is independent of the choice of deformation. In fact, it may be written as
    \begin{equation}
        \Del_a N^{a \dots}|_{\Sigma_0} = - \frac{1}{((k-1)!)^2} N^{c \dots} K_{cab} h^{ab}
        = \frac{1}{((k-1)!)^2} h^{a}{}_{b} \Del_a N^{b\dots}.
    \end{equation}
\end{corollary}

\begin{proof}
    In Theorem \ref{thm:kform-hodge}, we showed that the trace of the extrinsic curvature could be written in terms of the divergence of the unit $k$-normal form as
    \begin{equation} 
        K_{cab} h^{ab} 
            = - (-1)^s\, (k-1)!\, N_{c a_2 \dots a_k} \Del_{b} N^{b a_2 \dots a_k}.
    \end{equation}
    Contracting the free index $c$ with another copy of the unit $k$-normal form yields the expression
    \begin{equation}  \label{eq:hodge-hodge}
        N^{c d_2 \dots d_k} K_{cab} h^{ab} 
            = - (-1)^s\, (k-1)!\, (N^{c d_2 \dots d_k} N_{c a_2 \dots a_k}) \Del_{b} N^{b a_2 \dots a_k}.
    \end{equation}
    The contraction of two $k$-normal forms over a single index, being the contraction of two volume tensors, satisfies the identity\footnote{Again, see Appendix B of \cite{WaldBook} for a review of this statement.}
    \begin{equation} \label{eq:volume-form-identity}
        N^{c d_2 \dots d_k} N_{c a_2 \dots a_k}
            = (-1)^s (k-1)!\, P^{d_2}{}_{[a_2} \dots P^{d_k}{}_{a_k]}.
    \end{equation}
    Plugging this into equation \eqref{eq:hodge-hodge} yields the expression
    \begin{equation}
        N^{c d_2 \dots d_k} K_{cab} h^{ab} 
            = - ((k-1)!)^2\, P^{d_2}{}_{a_2} \dots P^{d_k}{}_{a_k} \Del_{b} N^{b a_2 \dots a_k},
    \end{equation}
    where we have exploited the antisymmetry of $\mathbf{N}$ to eliminate the asymmetrization coming from equation \eqref{eq:volume-form-identity}. Since the divergence of $\mathbf{N}$ is a $(k-1)$-form on the normal bundle, as per Lemma \ref{lem:div-knormal-normal}, it is unchanged under projection onto the normal bundle, and so the above expression can be rewritten as
    \begin{equation} \label{eq:appendix-final}
        N^{c d_2 \dots d_k} K_{cab} h^{ab} 
            = - ((k-1)!)^2\, \Del_{b} N^{b d_2 \dots d_k}.
    \end{equation}
    We may solve for the spacetime divergence of $\mathbf{N}$ to yield the expression
    \begin{equation} \label{eq:appendix-final-2}
        \Del_{b} N^{b d_2 \dots d_k}
            = - \frac{1}{((k-1)!)^2} N^{c d_2 \dots d_k} K_{cab} h^{ab}.
    \end{equation}
    By substituting equation \eqref{eq:trace-interm} for the trace of the extrinsic curvature, we may equivalently write the divergence of $\mathbf{N}$ as
    \begin{equation} \label{eq:appendix-final-3}
        \Del_{b} N^{b d_2 \dots d_k}
            = \frac{1}{((k-1)!)^2} h^{a}{}_{b} \Del_a N^{b\dots}.
    \end{equation}
    
    Equation \eqref{eq:appendix-final-2} expresses the divergence of $\mathbf{N}$ in terms of $\mathbf{N}$ on $\Sigma_0$ and the trace of the extrinsic curvature on $\Sigma_0.$ Equation \eqref{eq:appendix-final-3} expresses the divergence of $\mathbf{N}$ in terms of directional derivatives of $\mathbf{N}$ along directions tangent to $\Sigma_0.$ Both of these expressions depend only on the surface $\Sigma_0$ itself, and have no dependence on the choice of deformation.
\end{proof}

\vspace{0.4cm}

With Theorem \ref{thm:kform-hodge}, we have demonstrated equivalence between our ``differential forms'' characterization of extremality, defined by vanishing divergence of the unit $k$-normal field, and the commonly used extremality condition of vanishing trace of the extrinsic curvature. Each of the four equivalent extremality conditions we have discussed in this appendix is useful in different circumstances, and can be applied as needed depending on the problem at hand. In particular, we claim that the divergence of the unit $k$-normal form is particularly good for studying the asymptotic structure of extremal surfaces, as it is a differential equation that takes a relatively simple form in asymptotic coordinates. In the main body of this paper, this philosophy is put to work to prove cutoff-covariance of holographic entanglement entropy. We hope, however, that these techniques will find use beyond proving covariance of holographic entropy cutoffs; the techniques detailed in this appendix constitute a robust toolkit for the local, covariant analysis of extremal surfaces, and may be useful anywhere extremal surfaces are studied.

\bibliographystyle{JHEP}
\bibliography{cutoffs}

\end{document}